\documentclass[aps,pra,twocolumn,nofootinbib, superscriptaddress,10pt]{revtex4-1}
\usepackage{graphicx}
\usepackage{epstopdf}
\usepackage{amsmath}
\usepackage{amssymb}
\usepackage{mathrsfs}
\usepackage{amsthm}
\usepackage{bm}
\usepackage{url}
\usepackage[T1]{fontenc}
\usepackage{csquotes}
\MakeOuterQuote{"}

\usepackage{color}

\newtheoremstyle{note}
  {\topsep/2}               % ABOVE SPACE
  {\topsep/2}               % BELOW SPACE
  {}                      % BODY FONT
  {\parindent}            % INDENT (empty value is the same as 0pt)
  {\itshape}              % HEAD FONT
  {.}                     % HEAD PUNCTUATION
  {5pt plus 1pt minus 1pt}% HEAD SPACE
  {}

\theoremstyle{note}
\newtheorem{theorem}{Theorem}
\newtheorem{lemma}{Lemma}

\newtheorem{proposition}{Proposition}

\theoremstyle{definition}

\theoremstyle{remark}

\newcommand{\tr}{\operatorname{tr}}

\newcommand{\rk}{\operatorname{rank}}

 \newcommand{\rme}{\mathrm{e}}
 \newcommand{\rmi}{\mathrm{i}}

 \newcommand{\rmC}{\mathrm{C}}

 \newcommand{\na}{\mathrm{NA}}

 \newcommand{\bbZ}{\mathbb{Z}}

 \newcommand{\supp}{\mathrm{supp}}

 \newcommand{\id}{1}

  \newcommand{\scrA}{\mathscr{A}}
 \newcommand{\scrI}{\mathscr{I}}
 \newcommand{\scrN}{\mathscr{N}}

 \newcommand{\dfe}{\mathrm{DFE}}
 \newcommand{\mth}{\mathrm{MTH}}

\newcommand{\be}{\begin{equation}}
\newcommand{\ee}{\end{equation}}
\newcommand{\ba}{\begin{align}}
\newcommand{\ea}{\end{align}}

\def\<{\langle}  %% overriding the original command \<
\def\>{\rangle}  %% overriding the original command \>

\def\eqref#1{\textup{(\ref{#1})}}  %% overriding the original command \eqref
\newcommand{\eref}[1]{Eq.~\textup{(\ref{#1})}}
\newcommand{\Eref}[1]{Equation~\textup{(\ref{#1})}}
\newcommand{\esref}[1]{Eqs.~\textup{(\ref{#1})}}

\newcommand{\fref}[1]{Fig.~\ref{#1}}

\newcommand{\tref}[1]{Table~\ref{#1}}

\newcommand{\sref}[1]{Sec.~\ref{#1}}
\newcommand{\Sref}[1]{Section~\ref{#1}}

\newcommand{\thref}[1]{Theorem~\ref{#1}}
\newcommand{\Thref}[1]{Theorem~\ref{#1}}

\newcommand{\lref}[1]{Lemma~\ref{#1}}
\newcommand{\Lref}[1]{Lemma~\ref{#1}}
\newcommand{\lsref}[1]{Lemmas~\ref{#1}}
\newcommand{\Lsref}[1]{Lemmas~\ref{#1}}

\newcommand{\pref}[1]{Proposition~\ref{#1}}

\newcommand{\cref}[1]{Conjecture~\ref{#1}}
\newcommand{\Cref}[1]{Conjecture~\ref{#1}}

\newcommand{\aref}[1]{Appendix~\ref{#1}}

\newcommand{\rcite}[1]{Ref.~\cite{#1}}
\newcommand{\rscite}[1]{Refs.~\cite{#1}}

\begin{document}
\title{Efficient Verification of  Hypergraph States}

\author{Huangjun Zhu}
\email{zhuhuangjun@fudan.edu.cn}

\affiliation{Department of Physics and Center for Field Theory and Particle Physics, Fudan University, Shanghai 200433, China}

\affiliation{State Key Laboratory of Surface Physics, Fudan University, Shanghai 200433, China}

\affiliation{Institute for Nanoelectronic Devices and Quantum Computing, Fudan University, Shanghai 200433, China}

\affiliation{Collaborative Innovation Center of Advanced Microstructures, Nanjing 210093, China}

\affiliation{Institute for Theoretical Physics, University of Cologne,  Cologne 50937, Germany}

\author{Masahito Hayashi}
\affiliation{Graduate School of Mathematics, Nagoya University, Nagoya, 464-8602, Japan}

\affiliation{Shenzhen Institute for Quantum Science and Engineering,
Southern University of Science and Technology,
%No.1088 Xueyuan Avenue,Nanshan District,
Shenzhen,
518055, China}

\affiliation{Center for Quantum Computing, Peng Cheng Laboratory, Shenzhen 518000, China}

\affiliation{Centre for Quantum Technologies, National University of Singapore, 3 Science Drive 2, 117542, Singapore}

\begin{abstract}
	
Graph states and hypergraph states are of wide interest in quantum information processing and foundational studies. Efficient verification of these states  is a key to various applications. Here we propose a simple method for verifying hypergraph states which requires only two distinct Pauli measurements for each party, yet its efficiency is comparable to the best  strategy based on entangling measurements. For a given  state, the overhead is bounded by the chromatic number and degree of the underlying hypergraph. Our protocol is dramatically more efficient than all previous protocols based on local measurements, including tomography and direct fidelity estimation. It enables the verification of hypergraph states and genuine multipartite entanglement of thousands of qubits. The protocol can also be generalized to the adversarial scenario, while achieving almost the same efficiency. This merit is  particularly appealing to demonstrating  blind measurement-based quantum computation and quantum supremacy.
\end{abstract}

\date{\today}
\maketitle

\section{Introduction}
Entanglement is the characteristic feature of quantum theory and a key resource in quantum information processing \cite{HoroHHH09,GuhnT09}. As an archetypal example of quantum states with genuine multipartite entanglement (GME),
graph states  are of wide interest to  (blind) quantum computation \cite{RausB01,RausBB03,BroaFK09,MoriF13,HayaM15,FujiH17,HayaH18}, quantum error correction \cite{Gott97the,SchlW01},  quantum networks \cite{PersLCL13,MccuPBM16,MarkK18},
and foundational studies on nonlocality \cite{GreeHSZ90,ScarASA05,GuhnTHB05}.
Hypergraph states \cite{KrusK09,QuWLB13,RossHBM13,SteiRMG17,XionZCC18}, as a generalization of graph states, are equally useful in these research areas \cite{MillM16,MillM18,MoriTH17,GachBG16,GachGM19,Yosh16}. Moreover, certain hypergraph states, like Union Jack states, are universal for measurement-based quantum computation (MBQC) under only Pauli measurements \cite{MillM16,MillM18,GachGM19,TakeMH19}, which is impossible for graph states. Furthermore, hypergraph states are  attractive for demonstrating quantum supremacy \cite{BremMS16,MoriTH17} among other merits.

The applications of hypergraph states rely crucially on our ability to verify them with local measurements that are accessible in the lab.   However, no efficient method is known so far for verifying general hypergraph states, except for
graph states \cite{HayaM15,FujiH17,HayaH18, PallLM18,MarkK18,TakeMMM19}. In general, the resource required in traditional tomography increases exponentially with the number of qubits. The same is true for popular alternatives, such as compressed sensing \cite{GrosLFB10} and direct fidelity estimation (DFE) \cite{FlamL11}.   Even recent approaches tailored for hypergraph states  \cite{MoriTH17,TakeM18} are  too prohibitive to apply in  practice. The situation is much worse in the adversarial scenario, which  is crucial to many tasks in quantum information processing  that require high security conditions, including blind MBQC \cite{MoriF13,HayaM15,FujiH17,HayaH18,TakeMH19} and quantum networks \cite{PersLCL13,MccuPBM16,MarkK18}.
In this case, unfortunately,  to verify the simplest nontrivial hypergraph states (say of three qubits) already entails an astronomical number of measurements \cite{MoriTH17,TakeM18}.

Here we propose a simple and efficient method for verifying general (qubit and qudit) hypergraph states which requires only two distinct Pauli measurements for each party.
To verify an $n$-qubit hypergraph state, our protocol requires at most $m=n$ (potential) measurement settings and $ m\epsilon^{-1}\ln\delta^{-1}$ tests in total, where $\epsilon$ and $\delta$  denote the infidelity and  significance level, which characterize the target precision.
For a given state, $m$ can be replaced by the chromatic number or degree of the underlying hypergraph.   For many interesting graph states and hypergraph states, including cluster states and Union Jack states, the number of measurement settings and that of  tests in total do not increase with the number of qubits. For example,  Union Jack states can be verified with only three measurement settings and $3\epsilon^{-1}\ln\delta^{-1}$ tests in total.

Our protocol for verifying hypergraph states is dramatically more efficient than known candidates, including tomography and DFE \cite{FlamL11}, as well as recent  protocols tailored for hypergraph states \cite{MoriTH17,TakeM18}.
It enables efficient verification of hypergraph states  of thousands of  qubits
and is also highly efficient in certifying GME.
Moreover, our protocol can be generalized to the  adversarial scenario, while retaining almost the same efficiency, in which case
 the advantage over  previous approaches  is even more dramatic.  For stabilizer states (equivalent  to graph states under local Clifford transformations \cite{Schl02,GrasKR02}), our protocol adapted for the adversarial scenario is also much more efficient than previous protocols applicable to the adversarial scenario. Our proposal  is thus particularly appealing  to realizing blind MBQC and quantum supremacy.

The rest of this paper is organized as follows. In \sref{sec:QSV} we review the basic framework of pure-state verification.  In \sref{sec:HGS}, we review hypergraphs and hypergraph states  in preparation for later study. In \sref{sec:VHGS} we propose a simple and efficient protocol for verifying general  hypergraph states in  the nonadversarial scenario. In \sref{sec:VHGSadv} we generalize our protocol to the adversarial scenario and discuss an application to demonstrating quantum supremacy.
In \sref{sec:HGSGME} we provide efficient protocols for certifying GME of hypergraph states based on the above verification protocols. In \sref{sec:QuditHGS} we generalize our results to qudit hypergraph states. \Sref{sec:summary} summarizes this paper. To streamline the presentation, a few technical proofs and some additional discussions are relegated to the appendix.

\bigskip

\section{\label{sec:QSV}Verification of pure states}
Consider a device that is supposed to  produce the target state $|\Psi\>$ in the  Hilbert space ${\cal H}$.
In practice, the device may actually produce $\sigma_1, \sigma_2, \ldots, \sigma_{N}$ in $N$ runs.
Now our task is to
determine whether the average infidelity $\bar{\epsilon}=1-\bigr(\sum_j \<\Psi|\sigma_j|\Psi\>/N\bigr)$ of these states with the target state is smaller than a given threshold, say $\epsilon$.
To achieve this task we can perform two-outcome tests $\{P_l, 1-P_l\}$ based on local projective measurements. The test projector $P_l$ corresponds to passing the test and satisfies the condition $P_l|\Psi\>=|\Psi\>$, so that the target state can pass the test for sure. Suppose the test $P_l$ is performed with probability $\mu_l$, then the passing probability of a general state $\sigma$ is given by $\tr(\Omega\sigma)$, where
$\Omega:=\sum_l \mu_l P_l$ is the  verification operator and also called a strategy.
This probability satisfies the following equation \cite{PallLM18, ZhuH19AdL},
\begin{equation}\label{eq:PassingProb}
\max_{\<\Psi|\sigma|\Psi\>\leq 1-\epsilon }\tr(\Omega \sigma)=1- [1-\beta(\Omega)]\epsilon=1- \nu(\Omega)\epsilon,
\end{equation}
where $\beta(\Omega)$ is the second largest eigenvalue of $\Omega$, and  $\nu(\Omega):=1-\beta(\Omega)$ is the spectral gap.

\Eref{eq:PassingProb} implies that
the probability that all states $\sigma_1, \sigma_2, \ldots, \sigma_{N}$
pass the tests is at most
 $[1-\nu(\Omega)\bar{\epsilon}]^N$, where $\bar{\epsilon}$ is the average infidelity. To ensure the condition
$\sum_j \<\Psi|\sigma_j|\Psi\>/N> 1-\epsilon$ with significance level $\delta$, the minimum number of tests reads
\cite{ZhuH19AdS,ZhuH19AdL}
\begin{equation}\label{eq:NumTest}
\!N_\na(\epsilon,\delta,\Omega)\!=\!\biggl\lceil
\frac{1}{\ln[1-\nu(\Omega)\epsilon]}\ln\delta\biggr\rceil
\leq
\biggl\lceil
\frac{1}{\nu(\Omega)\epsilon}\ln\frac{1}{\delta}\biggr\rceil,
\end{equation}
where NA in the subscript means nonadversarial.
A similar formula was first derived in \rcite{PallLM18} under the assumption that
the fidelity $\<\Psi|\sigma_j|\Psi\>$ either equals 1  for all $j$ or satisfies $\<\Psi|\sigma_j|\Psi\>\leq 1- \epsilon$
for all $j$. Here we do not need this assumption.

The verification strategy $\Omega$ can also be applied to constructing upper and lower bounds for the average fidelity.
Suppose all  $\sigma_j$ are identical to the state $\sigma$ and   let $F=\<\Psi|\sigma|\Psi\>$ be the fidelity between $\sigma$ and the target state $|\Psi\>$. Then it is easy to derive that
\begin{equation}\label{eq:FidPassingProb1}
[1-\tau(\Omega)]F+\tau(\Omega) \leq \tr(\Omega \sigma)\leq \nu(\Omega)F+\beta(\Omega),
\end{equation}
where $\tau(\Omega)$ is the smallest eigenvalue of $\Omega$. As an implication, we have \cite{ZhuH19AdL}
\begin{equation}\label{eq:FidPassingProb2}
1-\tr(\Omega\sigma)\leq \frac{1-\tr(\Omega\sigma)}{1-\tau(\Omega)}\leq 1-F\leq\frac{1-\tr(\Omega\sigma)}{\nu(\Omega)} .
\end{equation}
In this way we can construct an interval in which the infidelity (or fidelity) lies by virtue of  the passing probability $\tr(\Omega \sigma)$. When $\sigma_j$ varies over different runs, \esref{eq:FidPassingProb1} and \eqref{eq:FidPassingProb2} are still applicable if $F$ and $\tr(\Omega \sigma)$ are replaced by their averages over all runs.

In the adversarial scenario, the quantum states are prepared by a potentially malicious adversary, as encountered in blind MBQC. Nevertheless, we can still verify pure quantum states efficiently by performing random permutations before applying a verification strategy $\Omega$ as in the nonadversarial scenario. In this case, the number of required tests is determined in \rscite{ZhuH19AdS,ZhuH19AdL}. Surprisingly, by adding the trivial test with a suitable probability, this number is comparable to the counterpart for the nonadversarial scenario, and the overhead is at most three times for high-precision verification. Therefore, to construct an efficient verification protocol for the adversarial scenario, it remains to devise
an efficient verification protocol for the nonadversarial scenario.

\section{\label{sec:HGS}Hypergraphs and hypergraph states}
\subsection{Hypergraphs}
 A hypergraph $G=(V,E)$ is characterized by  a set of vertices $V=\{1,2,\ldots,n\}$ and a set of hyperedges $E\subset \mathscr{P}(V)$, where $\mathscr{P}(V)$ is the power set of $V$ \cite{QuWLB13,RossHBM13}.
 The order of a hyperedge is the number of vertices it connects, and the order of a hypergraph is the maximum order of its hyperedges. As examples, \fref{fig:hypergraph} shows the
  order-3 hypergraphs underlying order-3 cluster states and Union Jack states \cite{MillM16}.
A graph is a special hypergraph in which all hyperedges have order 2.
Two distinct vertices of $G$ are adjacent if they are connected by a hyperedge. The degree $\deg(j)$ of a vertex $j$ is the number of vertices that are adjacent to it; the degree $\Delta(G)$ of $G$ is
the maximum vertex degree. A subset of $V$ is a clique if every  two vertices in the set are adjacent. The clique number $\varpi(G)$ of $G$ is the maximum number of vertices over all cliques. By contrast,
a subset is an independent set if no two vertices are adjacent. The independence number $\alpha(G)$ of $G$ is the  maximum number of vertices over all independent sets.

\begin{figure}
	\includegraphics[width=8.6cm]{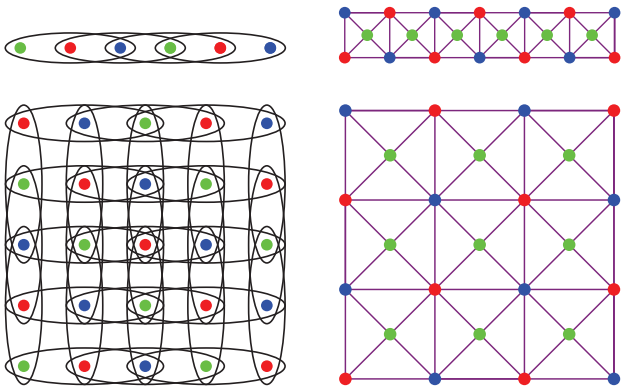}
	\caption{\label{fig:hypergraph}(color online) Examples of hypergraphs and associated hypergraph states. Left plot: 1D and 2D order-3 cluster states;
		every three neighboring vertices on a row or column are connected by an order-3 hyperedge. Right plot: Union Jack states on a chain and  a 2D lattice, respectively;	
		the three vertices of each elementary triangle are connected by an order-3 hyperedge \cite{MillM16}. All four hypergraphs are 3-colorable as illustrated.}
\end{figure}

A set  $\mathscr{A}=\{A_1, A_2, \ldots, A_m \}$ of independent sets of $G$ is an independence cover  if $\cup_{l=1}^m A_l=V$. The cover $\mathscr{A}$ defines a coloring of $G$ with $m$ colors when $\mathscr{A}$ forms a partition of $V$, that is,
when $A_l$ are pairwise disjoint (assuming no $A_l$ is empty).  A hypergraph $G$ is $m$-colorable if its vertices can be colored using $m$ different colors such that  adjacent vertices have different colors.  A bipartite graph is a 2-colorable graph.  The chromatic number $\chi(G)$  is the minimum number of colors in any coloring of $G$.

\begin{table}
	\caption{\label{tab:Graph}Degrees $\Delta(G)$, clique numbers $\varpi(G)$, independence numbers $\alpha(G)$, chromatic numbers $\chi(G)$, and independence degrees $\gamma(G)$ of common graphs and hypergraphs of $n$ vertices. A graph is complete if every two vertices are adjacent. Note that the odd cycle of three vertices is complete.
		Here we assume that each 2-colorable graph has at least one edge, while each 3-colorable hypergraph has at least one hyperedge of order 3, as illustrated in  \fref{fig:hypergraph}.  }
	\begin{tabular}{p{18ex}|ccccc}
		\hline\hline		
		Hypergraphs $G$ &  $\Delta(G)$ &  $\varpi(G)$ & $\alpha(G)$ & $\chi(G)$ & $\gamma(G)$ \\ \hline
		Square lattice  & 4 & 2 & $\lceil n/2\rceil$ &2 & 1/2\\
		Cubic lattice in dimension $k$ & $2k$ & 2 & $\lceil n/2\rceil$ &2 & 1/2\\
		Triangular lattice  & 6 & 3 & $\geq n/3$ &3 & 1/3\\
		
		Even cycle  & 2 & 2 & $ n/2$ &2 & 1/2\\
		Odd cycle($n\geq5$)  & 2 & 2 & $(n-1)/2$ &3 & $(n-1)/(2n)$\\
		Complete graph  & $n-1$ & $n$ & 1 &$n$ & $1/n$\\
		2-colorable graph  & - & 2 & $\geq n/2$ &2 & $1/2$\\
		3-colorable \hspace{3ex} hypergraph  & - & 3 & $\geq n/3$ &3 & $1/3$\\
		\hline\hline		
	\end{tabular}
\end{table}

 A weighted independence cover $(\mathscr{A},\mu)$ of $G$  is a cover together with weights $\mu_l$ for  $A_l\in \mathscr{A}$, where  $\mu_l$ form a probability distribution.
It  is closely connected to a fractional coloring \cite{HiltRS73,ScheU97} as explained in \aref{sec:HypergraphApp}.
 The cover strength of $(\mathscr{A},\mu)$  is defined as
\begin{equation}
s(\mathscr{A},\mu)=\min_{j\in V}\sum_{l|A_l \ni j} \mu_l.
\end{equation}
The independence degree $\gamma(G)$ of $G$ is the
maximum of $s(\mathscr{A},\mu)$ over all weighted independence covers of $G$, namely,
\begin{equation}
\gamma(G)=\max_{(\mathscr{A},\mu)} s(\mathscr{A},\mu).
\end{equation}
It  is also equal to the inverse of the fractional chromatic number $\chi_f(G)$, that is, $\gamma(G)=1/\chi_f(G)$ \cite{HiltRS73,ScheU97}.  The degrees, clique numbers, independence numbers, chromatic numbers, and independence degrees of common graphs and hypergraphs are shown in \tref{tab:Graph}.
The following  well-known proposition clarifies the relations among these hypergraph invariants; see  \aref{sec:HypergraphApp} for a self-contained proof and additional discussions.
\begin{proposition}\label{pro:IndependenceDegree}
Any hypergraph $G=(V,E)$ satisfies
\begin{equation}
\frac{1}{\Delta(G)+1}\leq \frac{1}{\chi(G)}\leq \gamma(G)\leq \min\left\{\frac{\alpha(G)}{|V|},\frac{1}{\varpi(G)}\right\}.
\end{equation}
\end{proposition}
As an implication of  \pref{pro:IndependenceDegree}, $\gamma(G)\geq 1/n$ for any hypergraph of $n$ vertices since $\Delta(G)\leq n-1$ and $\chi(G)\leq n$.
In addition, $\gamma(G)=1/m$ if
the  hypergraph $G$ has chromatic number and clique number both equal to $m$. In particular,  $\gamma(G)$ can attain the maximum 1 iff $G$ has no nontrivial hyperedges. Here a  hyperedge is nontrivial if its order	is larger than or equal to 2.
Any 2-colorable graph $G$ with at least one nontrivial edge has $\gamma(G)=1/2$.
For example $\gamma(G)=1/2$ when $G$ is a
square lattice  (or analogs in higher dimensions) or  an even cycle; $\gamma(G)=1/3$
when   $G$ is a triangular lattice.

\subsection{Hypergraph states}
The Pauli group for a qubit is generated by the two Pauli matrices
\begin{equation}
 X:=\biggl(\begin{matrix}
0 &1 \\
1&0
\end{matrix}\biggr),\quad  Z:=\biggl(\begin{matrix}
1 &0 \\
0&-1
\end{matrix}\biggr).
\end{equation}
The Pauli matrices for the $j$th qubit are
indexed by the subscript $j$.
Given any hypergraph $G=(V,E)$ with $n$ vertices, we can construct an $n$-qubit hypergraph state $|G\>$: prepare the state $|+\>=(|0\>+|1\>)/\sqrt{2}$ (eigenstate of $X$ with eigenvalue 1) for each vertex of $G$ and apply the generalized controlled-$Z$ operation $CZ_e$ on the vertices of each hyperedge $e\in E$ \cite{QuWLB13,RossHBM13},
that is,
\begin{equation}
|G\>=\Biggl(\prod_{e\in E} CZ_e\Biggr)|+\>^{\otimes n}.
\end{equation}
Here $CZ_e =\bigotimes_{j\in e} \id_j-2\bigotimes_{j\in e} |1\>\<1|_{j}$,
which acts trivially on the Hilbert space associated with vertices in $V\setminus e$.
When the  hyperedge $e$ is empty, $CZ_e$ is equal to the minus identity  $-1$ by convention.
When $e$ contains a single vertex, $CZ_e$ reduces to the Pauli operator $Z$ on the vertex, which is local.
When $e$ contains two vertices, $CZ_e$ is the familiar controlled-$Z$ operation.

Alternatively, the hypergraph state $|G\>$
 is the unique eigenstate (up to a global phase factor) with eigenvalue 1  of  the following $n$ commuting (nonlocal) stabilizer operators \cite{QuWLB13,RossHBM13},
\begin{equation}\label{eq:StabilizerHGS}
K_j =X_j\otimes \prod_{e\in E|\, e\ni j	} CZ_{e\setminus \{j\}}, \quad j=1,2,\ldots, n.
\end{equation}
This alternative characterization plays a key role in the verification of hypergraph states.

The order of a hypergraph state is defined as the order of the underlying hypergraph; similar convention  applies to many other graph theoretic quantities, such as the degree, clique number, (fractional) chromatic number, independence number, and independence degree.
For example $\gamma(G)=1/2$ for  graph states $|G\>$ associated with  nontrivial 2-colorable graphs (with at least one edge), including
cluster states (of any dimension); $\gamma(G)=1/3$ for  hypergraph states $|G\>$ associated with nontrivial
3-colorable hypergraphs (with at least one hyperedge of order 3): including
order-3 cluster states (of any dimension) and  Union Jack states; cf.~\tref{tab:Graph}.

Hypergraph states enjoy a number of appealing merits.
For example,  hypergraph states of connected hypergraphs are genuinely multipartite entangled (GME)  \cite{QuWLB13}. Certain hypergraph states, including Union Jack states  shown in \fref{fig:hypergraph}, are universal for MBQC under only Pauli measurements \cite{MillM16,MillM18,GachGM19}, which is impossible for graph states. What is more appealing, hypergraph states are found recently that are universal for MBQC under only $X$ and $Z$  measurements \cite{TakeMH19}. In addition, certain hypergraph states possess symmetry-protected topological orders,  which are a focus of  ongoing research \cite{MillM16,MillM18,Yosh16}. Furthermore, hypergraph states are attractive for demonstrating quantum supremacy \cite{BremMS16,MoriTH17}.
When $G$ is an ordinary graph, $|G\>$ reduces to a graph state. All stabilizer states  are equivalent to graph states under local Clifford transformations (LC) \cite{Schl02,GrasKR02}. Meanwhile,
Calderbank-Shor-Steane  states
are equivalent to  graph states associated with  2-colorable graphs, and vice versa \cite{ChenL07}.

\section{\label{sec:VHGS}Efficient verification of hypergraph states}

\subsection{Construction of tests for hypergraph states}  Let $G=(V,E)$ be a hypergraph with $n$ vertices and $|G\>$ the associated hypergraph state.
Given any nonempty independent set $A$  of $G$,
we can devise a test for $|G\>$ based on two types of Pauli measurements. The test consists in measuring  $X_j$ for  $j\in A$ and measuring $Z_k$ for $k\in \overline{A}$, where $\overline{A}:=V\setminus A$ is the complement of $A$ in $V$.   The
measurement outcome on the $a$th qubit for $a=1,2,\ldots, n$ can be written as  $(-1)^{o_a}$, where the Boolean variable $o_a$ is either 0 or 1.
Since $A$ is an independent set,
 $X_j$ and $Z_k$ commute with $K_i$ for all $i,j\in A$ and $k\in \overline{A}$. The joint eigenstate of $X_j$ and $Z_k$ corresponding to the outcome $\{o_a\}$
is  an eigenstate of $K_i$ with eigenvalue  $(-1)^{t_i}$, where
\begin{equation}\label{eq:EigenvalueStab}
t_i=o_i+ \sum_{e\in E| e \ni i} \Biggl(\prod_{k\in e, k\neq i}o_k\Biggr).
\end{equation}
Here it is understood that $\prod_{k\in e, k\neq i}o_k=1$ if $e=\{i\}$.

Now we set the criterion that
the test is passed iff $(-1)^{t_i}=1$ for all $i\in A$, then the test effectively measures all the stabilizer operators $K_i$ for $i\in A$.
The projector onto the pass eigenspace reads
\begin{equation}\label{eq:PassProj}
P_A=\prod_{i\in A}\frac{1+K_i}{2}.
\end{equation}
A quantum state $\rho$ can always pass the test iff it is stabilized by
$K_i$  for all $i\in A$, which holds for the target state $|G\>$. The rank of the test projector $P_A$  reads
\begin{equation}
\rk(P_A)=\tr(P_A)=2^{n-|A|},
\end{equation}
where $|A|$ denotes the cardinality of the set $A$; the larger is $|A|$, the smaller is  $\rk(P_A)$. In view of this observation, it is beneficial to choose large independent sets for constructing test projectors for the hypergraph state $|G\>$. Incidentally, the cardinality $|A|$ is upper bounded by the independence number $\alpha(G)$.   Suppose $G$ has at least one nontrivial hyperedge or edge; then $\alpha(G)\leq n-1$, which implies that $\rk(P_A)\geq 2$. So at least two distinct tests are necessary to verify $|G\>$ as expected.

Denote by $\scrN(A)$ the neighborhood of $A$ in the graph $G$, that is, the set of vertices in $G$ that are adjacent to at least one vertex in $A$.
Since  $A$ is an independent set,  $\scrN(A)$  and $A$ are disjoint, that is, $\scrN(A)\subset \overline{A}$.
The complement $\overline{A}$ used in constructing the  test $P_A$ above  can  be replaced by $\scrN(A)$ because \eref{eq:EigenvalueStab} only involves  measurement outcomes associated with  vertices in the set $A\cup \scrN(A)$. In other words,  $Z$ measurements associated with vertices in $\overline{A}\setminus\scrN(A)$ are redundant.  Meanwhile, the independent set $A$ can be enlarged when $A\cup\scrN(A)$ is a proper subset of the vertex set $V$ of $G$. Conversely, if $A\cup\scrN(A)=V$, then $\scrN(A)=\overline{A}$, so that $A$ cannot be contained in any larger independent set.

 \begin{figure*}
	\includegraphics[width=14cm]{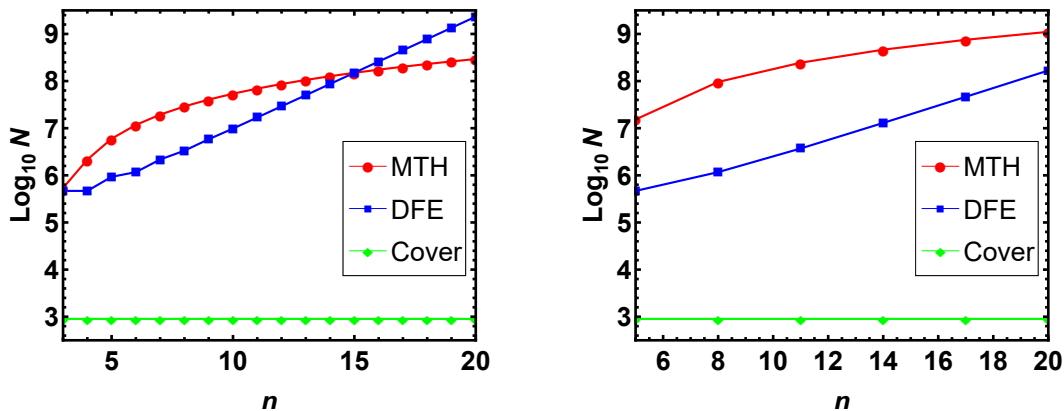}
	\caption{\label{fig:NumberMeas}(color online) Resource costs for verifying hypergraph states  in the nonadversarial scenario. Left plot: 1D order-3 cluster states; right plot: Union Jack states on a chain. Here $n$ is the number of qubits, and $N$ is the (expected) number of tests  required to verify the state within infidelity $\epsilon=0.01$ and significance level $\delta=0.05$.  In the case of the MTH protocol proposed in \rcite{MoriTH17}, only a lower bound for $N$ is given.  The lines are guides for the eye.
		Our cover protocol dramatically outperforms direct fidelity estimation (DFE) \cite{FlamL11}  and the MTH protocol (cf.~\aref{sec:Comparison}).
	}
\end{figure*}

\subsection{\label{sec:cover}The cover protocol}  Let $\mathscr{A}=\{A_1, A_2, \ldots, A_m \}$  be an independence cover  of $G$ that is composed of $m$ nonempty independent sets, then  we can devise  a verification protocol for $|G\>$ with $m$ distinct tests. For each independent set $A_l$, we can construct a test  with the test projector $P_l=\prod_{i\in A_l}\frac{1+K_i}{2}$ according to \eref{eq:PassProj}.
A state can pass all $m$ tests iff it is stabilized by  $K_i$  for all $i\in \cup_{l=1}^m A_l=V$. So only the target  state $|G\>$ can pass all tests with certainty as desired. This verification protocol is referred to as the \emph{cover protocol} (or fractional coloring protocol) since it is determined by an independence cover (or a fractional coloring). When $G$ is connected, the hypergraph state $|G\>$ is GME, so each party requires at least two distinct projective measurements to verify $|G\>$ \cite{ZhuH19AdL}. The cover protocol requires only two Pauli measurements for each party and is thus the most economical with regard to the number of measurement settings for each party.

Suppose the test $P_l$ (associated with the  independent set $A_l$) is applied with probability $\mu_l$. Then  the cover protocol is characterized by the weighted independence cover $(\mathscr{A},\mu)$.
Its efficiency  is determined by the spectral gap of the verification operator 
\begin{align}
\Omega(\mathscr{A},\mu)=\sum_{l=1}\mu_l P_l=\sum_l \mu_l \prod_{i\in A_l}\frac{1+K_i}{2}. 
\end{align}
Note that the  common eigenbasis of $K_i$ for $i\in V$ also forms an eigenbasis of $\Omega(\mathscr{A},\mu)$. Each eigenstate $|\Psi_x\>$ in this basis is specified by  an $n$ bit string  $x\in \{0,1\}^n$ and  satisfies the equation $K_i|\Psi_x\>=(-1)^{x_i}|\Psi_x\>$.  The corresponding eigenvalue of $\Omega(\mathscr{A},\mu)$ reads
\begin{equation}
\lambda_x= \sum_{l| \supp(x)\subset \overline{A}_l} \mu_l,
\end{equation}
where $\supp(x):=\{i\,|\, x_i\neq 0\}$.
To attain the second largest eigenvalue of $\Omega(\mathscr{A},\mu)$, it suffices to consider the case in which  $x$ has only one bit equal to 1,  which means
\begin{equation}
\beta(\Omega(\mathscr{A},\mu))=\max_{i\in V}\sum_{l| \overline{A}_l\ni i } \mu_l.
\end{equation}
Similarly, the smallest eigenvalue of $\Omega(\mathscr{A},\mu)$ is attained when all bits of $x$ are equal to 1, in which case we have $\lambda_x=0$, given that all independent sets $A_l$ are nonempty. So the verification operator $\Omega(\mathscr{A},\mu)$ is always singular.
These observations confirm the following theorem.
\begin{theorem}\label{thm:CoverEfficiency}The spectral gap and the smallest eigenvalue of the cover protocol read
\begin{align}
&\nu(\Omega(\mathscr{A},\mu)) =s(\mathscr{A},\mu),\quad \tau(\Omega(\mathscr{A},\mu)) =0,\\
&\max_{(\mathscr{A},\mu)}\nu(\Omega(\mathscr{A},\mu))=\gamma(G)=[\chi_f(G)]^{-1}.
\end{align}
\end{theorem}

When the independent sets $A_1, A_2, \ldots, A_m$ are pairwise disjoint, $\mathscr{A}$ defines a coloring of $G$, in which case the  protocol $(\mathscr{A},\mu)$  is also called
a \emph{coloring protocol}. Each test of the coloring protocol is associated with a color: $X$ measurement is performed on all qubits associated with a given color, while $Z$ measurement is performed on other  qubits.  The number of distinct tests equals the number of colors. For example, the Union Jack state shown in \fref{fig:hypergraph} can be verified using a coloring protocol composed of three distinct tests.
The spectral gap of the resulting verification operator $\Omega(\mathscr{A},\mu)$ reads
\begin{equation}\label{eq:vcoloring}
\nu(\Omega(\mathscr{A},\mu))=\min_{l} \mu_l\leq |\scrA|^{-1}\leq \chi(G)^{-1}.
\end{equation}
Here the first inequality is saturated iff all weights $\mu_l$ are equal; the second one is saturated iff $|\scrA|=\chi(G)$, so that the coloring $\scrA$ is optimal (cf.~\aref{sec:HypergraphApp}). In view of this observation, by a coloring protocol, we shall assume that  all weights $\mu_l$ are equal, that is, all distinct  tests are performed with the same probability. Then the coloring protocol $(\mathscr{A},\mu)$ is also denoted by $\mathscr{A}$.  Here we emphasize that it is not necessary to compute $\gamma(G)$ or $\chi(G)$ to apply a cover or coloring protocol. In practice, a good coloring can be found  using the Dsatur algorithm \cite{Brel79} for example.

\Thref{thm:CoverEfficiency} reveals  operational meanings of the cover strength, independence degree, and fractional chromatic number in  verifying a hypergraph state. Given a cover protocol $(\mathscr{A},\mu)$, to verify  $|G\>$ within infidelity $\epsilon$ and significance level $\delta$, the minimum number of tests  is
\begin{align}\label{eq:NumTestCover}
N=\biggl\lceil
\frac{\ln\delta}{\ln[1-s(\mathscr{A},\mu)\epsilon]}\biggr\rceil
\leq
\biggl\lceil
\frac{\ln\delta^{-1}}{s(\mathscr{A},\mu)\epsilon}\biggr\rceil
\end{align}
by \eref{eq:NumTest}.
This number is minimized when the cover $(\mathscr{A},\mu)$ is optimal, so that $s(\mathscr{A},\mu)=\gamma(G)=1/\chi_f(G)$ and
\begin{align}\label{eq:NumTestCoverOpt}
N=\biggl\lceil
\frac{\ln\delta}{\ln[1-\gamma(G)\epsilon]}\biggr\rceil
\leq
\biggl\lceil
\frac{\ln\delta^{-1}}{\gamma(G)\epsilon}\biggr\rceil= \biggl\lceil
\frac{\chi_f(G)\ln\delta^{-1}}{\epsilon}\biggr\rceil
.
\end{align}
According to \pref{pro:IndependenceDegree}, we have
\begin{equation}\label{eq:NumTestCoverUB}
\!\!N\leq \biggl\lceil\frac{\chi(G)}{\epsilon}\ln \frac{1}{\delta} \biggr\rceil\leq \biggl\lceil\frac{\Delta(G)+1}{\epsilon}\ln \frac{1}{\delta}\biggr\rceil\leq \biggl\lceil\frac{n}{\epsilon}\ln \frac{1}{\delta}\biggr\rceil.
\end{equation}
Here the first  upper bound  can be achieved by an optimal coloring protocol; the second one can be achieved by a coloring with $\Delta(G)+1$ colors.

Although,  in general, it is not easy to find an optimal coloring of the hypergraph $G$, it is easy
to find a coloring  with $m\leq \Delta(G)+1$ colors by virtue of a simple greedy algorithm  as presented in the proof of \pref{pro:IndependenceDegree}; cf.~\rcite{Brel79}. If we apply a coloring protocol with $m$ colors, then the spectral gap of the verification operator reads $\nu(\Omega)=1/m$, so 
the minimum number of tests  is given by 
\begin{align}
N=\biggl\lceil
\frac{\ln\delta}{\ln(1-m^{-1}\epsilon)}\biggr\rceil
\leq
\biggl\lceil
\frac{m\ln\delta^{-1}}{\epsilon}\biggr\rceil.
\end{align}

The above analysis shows that any hypergraph state $|G\>$ can be verified with at most $m=\Delta(G)+1$ measurement settings  in which each party performs either $X$ or $Z$ measurement. Note that $m$ is upper bounded by the number $n$ of qubits.
The total number of tests is only  $\lceil m\epsilon^{-1}\ln\delta^{-1}\rceil$ and
is at most $m$ times as large as the number for the best  protocol based on  entangling measurements.   The cover protocol for verifying hypergraph states is dramatically more efficient than previous protocols \cite{FlamL11,MoriTH17}, as illustrated in  \fref{fig:NumberMeas} and discussed in detail in  \aref{sec:Comparison}.
Consider the protocol of \rcite{MoriTH17} for example, both the number of measurement settings and the total number of tests increase exponentially with $\Delta(G)$; in addition, the number of tests scales as $1/\epsilon^2$ instead of $1/\epsilon$.

For many interesting  hypergraph states, the chromatic numbers do not grow with the qubit number. So  the number of measurement settings and the total number of tests required by  the coloring protocol are independent of the number of qubits, which is the same as the best protocol based on entangling measurements.
For example,  two measurement settings and $\lceil 2\epsilon^{-1}\ln \delta^{-1} \rceil$ tests are sufficient for graph states of 2-colorable graphs (equivalent to  Calderbank-Shor-Steane states \cite{ChenL07}), including GHZ states, cluster states (of arbitrary dimensions),  tree graph states, and graph states associated with even cycles.  Three settings and $\lceil 3\epsilon^{-1}\ln \delta^{-1} \rceil$ tests are sufficient for order-3 cluster states and Union Jack states (cf.~\fref{fig:hypergraph}).

Incidentally, the cover or coloring protocol can also be applied to constructing upper and lower bounds for the average infidelity between the states prepared and the target hypergraph state $|G\>$. Suppose all states $\sigma_j$ prepared in different runs are identical to $\sigma$.
If we apply the optimal cover protocol, then the spectral gap of the verification operator reads $\nu(\Omega)=\gamma(G)=1/\chi_f(G)$. According to \eref{eq:FidPassingProb2},  the infidelity between $\sigma$ and $|G\>$ satisfies
\begin{equation}
1-\tr(\Omega\sigma)\leq 1-\<G|\sigma|G\>\leq\chi_f(G)[1-\tr(\Omega\sigma)].
\end{equation}
If instead a coloring protocol with $m$ colors is applied, then $\nu(\Omega)=1/m$, so that
\begin{equation}
1-\tr(\Omega\sigma)\leq 1-\<G|\sigma|G\>\leq m[1-\tr(\Omega\sigma)].
\end{equation}
In general, if $\sigma_j$ are not identical with each other, the above conclusions still hold if
$\<G|\sigma|G\>$ and $\tr(\Omega\sigma)$ are replaced by their averages over all runs.

\section{\label{sec:VHGSadv}Verification of hypergraph states in the adversarial scenario}
\subsection{Cover protocol for the adversarial scenario}
Thanks to a general recipe proposed in \rscite{ZhuH19AdS,ZhuH19AdL}, the cover protocol  can also be applied  to the adversarial scenario,  which is very important to many tasks in quantum information processing that demand high-security requirements, such as  blind MBQC.
Let $\Omega=\Omega(\scrA,\mu)$ be the verification operator associated with the cover protocol $(\scrA,\mu)$, then  $\nu(\Omega)=s(\mathscr{A},\mu)$ and  $\tau(\Omega)=0$ according to \thref{thm:CoverEfficiency}.  By
\rcite{ZhuH19AdL}, the number of tests required by  $\Omega$ to verify $|G\>$ within infidelity $\epsilon$ and significance level $\delta$ in  the adversarial scenario satisfies
\begin{align}\label{eq:NumTestAdvHG}
\min\left\{
\biggl\lceil\frac{1- \delta}{\nu(\Omega) \delta \epsilon}\biggr\rceil,\; \biggl\lceil\frac{1}{\delta\epsilon}-1\biggr\rceil\right\}\leq N\leq \biggl\lceil\frac{1-\delta}{\nu(\Omega)\delta\epsilon} \biggr\rceil.
\end{align}
For the optimal coloring protocol with $\nu(\Omega)=1/\chi(G)$, we have
\begin{align}\label{eq:NumTestAdvHG2}
N\leq \biggl\lceil\frac{\chi(G)(1-\delta)}{\delta\epsilon} \biggr\rceil\leq \biggl\lceil\frac{\Delta(G)+1}{\delta\epsilon} \biggr\rceil\leq \biggl\lceil\frac{n}{\delta\epsilon} \biggr\rceil.
\end{align}
If in addition $G$ is  2-colorable, then $\nu(\Omega)=1/2$ and the lower bound in  \eref{eq:NumTestAdvHG} is saturated according to  \rcite{ZhuH19AdL}. The cover protocol adapted to the adversarial scenario is much more efficient than all previous protocols for verifying hypergraph states known in the literature. Nevertheless, the scaling of $N$ with $\delta$ is suboptimal because the verification operator $\Omega(\mathscr{A},\mu)$ is singular.

 \begin{figure}
	\includegraphics[width=7.5cm]{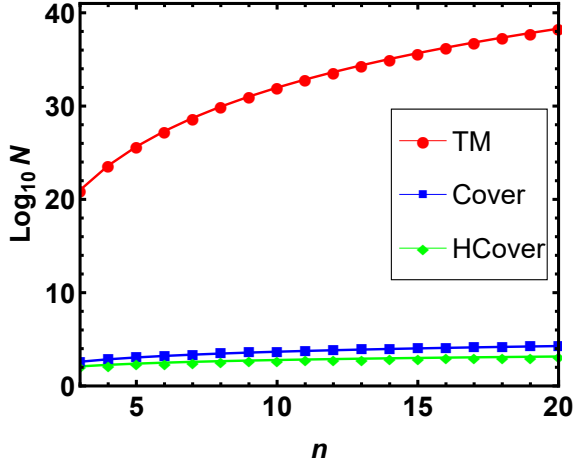}
	\caption{\label{fig:NumberAdv}(color online)  Resource costs for verifying 3-colorable hypergraph states  in the adversarial scenario. Here $n$ is the number of qubits, and $N$ is the number of tests required to verify the state within infidelity $\epsilon=1/(4n)$ and significance level $\delta=1/(4n)$. The lines are guides for the eye. Our cover protocol (Cover) and hedged cover protocol (HCover) outperform the TM protocol proposed in \rcite{TakeM18} by at least 18 orders of magnitude.
	}
\end{figure}

\subsection{\label{sec:HedgedCover}Hedged  cover protocol}
To improve the scaling behavior of the number of required tests with $\delta$, here we propose a  \emph{hedged  cover protocol} $(\mathscr{A},\mu)_p$, which is characterized by the  verification operator
\begin{equation}
\Omega_{p}=(1-p)\Omega +p
\end{equation}
 with $\Omega=\Omega(\mathscr{A},\mu)$. It can be realized by performing $\Omega$ with probability $1-p$ and the trivial test (the test projector is the identity) with probability $p$. The name "hedged  cover protocol" reflects the fact that the trivial test is introduced to hedge the influence of small eigenvalues of the verification operator $\Omega$ of the cover protocol.
 The second largest and smallest eigenvalues of $\Omega_p$ read
 \begin{equation}
 \beta_p=(1-p)\beta+p=1-\nu+p\nu,\quad \tau_p=p,
 \end{equation}
where $\beta$ is the second largest eigenvalue of $\Omega$.
The protocol $(\mathscr{A},\mu)_p$ is also called a hedged coloring protocol and denoted by  $\mathscr{A}_p$ when $\mathscr{A}$ denotes a coloring and all $\mu_l$ are equal.

According to \rscite{ZhuH19AdS,ZhuH19AdL}, to verify $|G\>$ within infidelity $\epsilon$ and significance level $\delta$ in the adversarial scenario, the number of tests required by the strategy $\Omega_p$ (with $p>0$) satisfies
\begin{equation}\label{eq:NumTestAdvRev}
N< \frac{h(p,\nu) \ln(F\delta)^{-1}}{\epsilon},
\end{equation}
where $\nu=s(\mathscr{A},\mu)$, $F=1-\epsilon$, and
\begin{align}
h(p,\nu)&=\bigl[\min\bigl\{\beta_p \ln \beta_p^{-1}, p \ln p^{-1}\bigr\}\bigr]^{-1}.
\end{align}	
To achieve high efficiency, the value of $p$ can be chosen as follows according to \rscite{ZhuH19AdS,ZhuH19AdL},
\begin{equation}\label{eq:pstardef}
p_*(\nu)=\min\bigl\{p> 0| p \ln p^{-1}\geq \beta_p \ln \beta_p^{-1} \bigr\}.
\end{equation}
With this choice, the hedged cover protocol  $(\mathscr{A},\mu)_p$ is also denoted by $(\mathscr{A},\mu)_*$; similarly, the hedged coloring protocol  $\mathscr{A}_p$ is  denoted by $\mathscr{A}_*$.
By virtue of the hedged cover protocol $(\mathscr{A},\mu)_*$, that is, $(\mathscr{A},\mu)_p$ with  $p=p_*(\nu)$, the number of tests in \eref{eq:NumTestAdvHG} can be reduced to \cite{ZhuH19AdS,ZhuH19AdL}
\begin{align}
N&=\biggl\lfloor \frac{h_*(\nu) \ln(F\delta)^{-1}}{\epsilon}\biggr\rfloor\leq  \frac{\ln(F\delta)^{-1}}{(1-\nu+\rme^{-1}\nu^2)\nu\epsilon}\nonumber\\
&\leq \frac{(1+\rme \nu-\nu)\ln(F\delta)^{-1}}{\nu\epsilon}\leq \frac{\rme\ln(F\delta)^{-1}}{\nu\epsilon}, \label{eq:NumTestAdvHGhc}
\end{align}
where $\nu=s(\mathscr{A},\mu)$, $F=1-\epsilon$,
\begin{equation}\label{eq:hstardef}
h_*(\nu)=h(p_*(\nu),\nu)=-[p_*(\nu)\ln p_*(\nu)]^{-1},
\end{equation}
and $\rme$ is the base of the natural logarithm.

Alternatively, we can choose
 $p=\nu/\rme$ to construct an efficient hedged cover protocol  \cite{ZhuH19AdS,ZhuH19AdL}.
It turns out that the three upper bounds in \eref{eq:NumTestAdvHGhc} still apply if the hedged cover protocol $(\mathscr{A},\mu)_*$ is replaced by $(\mathscr{A},\mu)_p$ with
$p=\nu/\rme$, that is,
\begin{align}
N&\leq  \frac{\ln(F\delta)^{-1}}{(1-\nu+\rme^{-1}\nu^2)\nu\epsilon}\leq \frac{(1+\rme \nu-\nu)\ln(F\delta)^{-1}}{\nu\epsilon}\nonumber\\
&\leq \frac{\rme\ln(F\delta)^{-1}}{\nu\epsilon}.
\end{align}

If $\mathscr{A}$ is  a coloring protocol with $m$ colors, then  the spectral gap reads $\nu=1/m$, and the number of tests required by the hedged coloring protocol $\mathscr{A}_*$  is given by \eref{eq:NumTestAdvHGhc} with $\nu=1/m$, that is,
\begin{align}\label{eq:NumTestAdvHGhcolor}
&N=\biggl\lfloor \frac{h_*(1/m)\ln(F\delta)^{-1}}{\epsilon}\biggr\rfloor\leq  \frac{(m+\rme-1)\ln(F\delta)^{-1}}{\epsilon}.
\end{align}
 If the coloring  is optimal, then we have $m=\chi(G)$ and $\nu=1/\chi(G)$, so that
\begin{align}\label{eq:NumTestAdvHGhcolorOpt}
&N=\biggl\lfloor \frac{h_*(1/\chi(G))\ln(F\delta)^{-1}}{\epsilon}\biggr\rfloor\leq  \frac{[\chi(G)+\rme-1]\ln(F\delta)^{-1}}{\epsilon}\nonumber\\
&\leq \frac{[\Delta(G)+\rme]\ln(F\delta)^{-1}}{\epsilon}\leq \frac{(n+\rme-1)\ln(F\delta)^{-1}}{\epsilon}.
\end{align}
Again, the  upper bound in \eref{eq:NumTestAdvHGhcolor} and the last three  bounds in \eref{eq:NumTestAdvHGhcolorOpt} still apply if the hedged coloring protocol $\mathscr{A}_*$  is replaced by  $\mathscr{A}_p$ with
$p=\nu/\rme$.
According to the above results, the hedged cover (or coloring)  protocol can achieve
the same optimal scaling behavior in the number $N$ of tests  with $\epsilon^{-1}$ and $\delta^{-1}$ as the counterpart in the nonadversarial scenario; cf.~\eref{eq:NumTestCoverUB}. For high-precision verification,  the overhead is at most three times in general and is even negligible when $\chi(G)$ is large.
Therefore, all the conclusions on the cover protocol presented in \sref{sec:cover} can easily be adapted to the adversarial scenario.

To illustrate the advantage of our approach,
suppose we want to verify a general $n$-qubit  hypergraph state within   infidelity $\epsilon$ and significance level $\delta=\epsilon$ in the adversarial scenario.   The protocol
in a recent paper \rcite{TakeM18} requires at least $(2\ln 2) n^3\epsilon^{-18}$ tests [applicable when $4n\epsilon\leq 1$, cf.  \eref{eq:NTM1} in \aref{sec:Comparison}]. By contrast, our coloring protocol  requires at most $\lceil n\epsilon^{-2}\rceil$ tests by \eref{eq:NumTestAdvHG}, note that $\nu(\Omega)\geq 1/n$ for any coloring protocol. The hedged coloring protocol requires about $n\epsilon^{-1}\ln \epsilon^{-1}$ tests if $n,\epsilon^{-1}\gg 1$. When $\delta=\epsilon=1/(4n)$ and the hypergraph is
3-colorable,  the protocol
in \rcite{TakeM18} requires at least $9.5\times 10^{10}n^{21}$ tests, which is  astronomical. By contrast, the optimal  cover or coloring protocol with $\nu(\Omega)=\gamma(G)=1/\chi(G)=1/3$ requires at most $12n(4n-1)$ tests by  \eref{eq:NumTestAdvHG}, which  outperforms \rcite{TakeM18} by at least 18 orders of magnitude even when $n=3$, and the advantage increases rapidly  with $n$, as illustrated in \fref{fig:NumberAdv}. The hedged cover or coloring protocol can further reduce the number to $\bigl\lfloor 16.3 n\ln\frac{16n^2}{4n-1}\bigr\rfloor$
by \eref{eq:NumTestAdvHGhcolorOpt}, given that $h_*(\nu=1/3)< 4.052$, which can be verified by straightforward numerical calculation.

\subsection{Application to demonstrating quantum supremacy}
Our protocols for verifying hypergraph states are  instrumental to demonstrating quantum supremacy.  Recently, on the basis of a plausible complexity theoretic assumption, Bremner,  Montanaro, and  Shepherd demonstrated the average-case hardness of sampling from probability distributions from instantaneous quantum polytime (IQP) circuits \cite{BremMS16} (Theorem~6 there). Their  work implies the hardness of sampling from probability distributions resulting from $X$ measurements on random order-3 hypergraph states such that the error in $\ell_1$-norm is bounded by $1/192$.

To guarantee that  the error in $\ell_1$-norm  is bounded by $1/192$ so as to demonstrate quantum supremacy, it suffices to ensure that the trace distance between the hypergraph state generated and the ideal target state is bounded by $1/192$.
Note that the trace distance $D_{\tr}(\rho_1,\rho_2)$ between two quantum states $\rho_1, \rho_2$  and the fidelity $F(\rho_1,\rho_2)$ satisfies
\begin{equation}
1-\sqrt{F(\rho_1,\rho_2)}\leq D_{\tr}(\rho_1,\rho_2)\leq \sqrt{1-F(\rho_1,\rho_2)},
\end{equation}
where the upper bound is saturated when $\rho_1$ and $\rho_2$ are pure \cite{NielC00book} (our definition of the fidelity is the square of the counterpart in \rcite{NielC00book}). To demonstrate quantum supremacy,
we need to verify the hypergraph state within infidelity $\epsilon=1/192^2$ and a given significance level $\delta$. If  $\delta=\epsilon$, then the number of tests required by the protocol in \rcite{TakeM18} is  more than $2\times 10^{82} n^3$ [applicable when $4n\epsilon\leq 1$, cf.~\eref{eq:NTM1}]. By contrast, our hedged coloring protocol requires only about $4\times 10^5 n$ tests according to \eref{eq:NumTestAdvHGhcolor} and  is thus dramatically more efficient.

\section{\label{sec:HGSGME}Certification  of genuine multipartite entanglement}
Here we show that the cover protocol and the hedged cover protocol are surprisingly efficient in certifying  GME of hypergraph states, although it is not necessarily optimized for this purpose. Recall that a multipartite pure state is GME if it is not biseparable, that is, if it cannot
be written as a tensor product of two pure states.
A mixed state is GME if it cannot be expressed as a convex mixture of biseparable states \cite{GuhnT09}.

\subsection{Nonadversarial scenario}
\begin{theorem}\label{thm:GenuineEnt}
	Let  $G$ be a connected order-$k$ hypergraph and $|G\>$ the corresponding hypergraph state. If a state $\rho$ satisfies $\<G|\rho|G\> > 1-2^{1-k}$,  then 	$\rho$ is GME.
\end{theorem}
This theorem was proved in \rcite{GhioMRB18}; see \aref{sec:GMEHGSproof} for an independent proof. Note that the conclusion is independent of the number  $n$  of qubits.
\Thref{thm:GenuineEnt} is known much earlier   when $|G\>$ is a graph state associated with a connected graph, in which case $\rho$ is GME if its  fidelity with $|G\>$ is  larger than one half \cite{TothG05, TothG05E,GuhnT09}.
In general, to certify the GME of the hypergraph state $|G\>$ with significance level $\delta$, we need to guarantee the fidelity $\<G|\rho|G\> > 1-2^{1-k}$ with significance level $\delta$.
Given a verification strategy $\Omega$, then
it suffices to perform
\begin{equation}\label{eq:NumTestGME}
N=\biggl\lceil
\frac{1}{\ln[1-2^{1-k}\nu(\Omega)]}\ln\delta\biggr\rceil\leq
\biggl\lceil
\frac{2^{k-1}}{\nu(\Omega)}\ln\delta^{-1}\biggr\rceil
\end{equation}
tests according to \eref{eq:NumTest} with $\epsilon=2^{1-k}$.

If $\Omega$ corresponds to the cover protocol $(\mathscr{A},\mu)$, then $\nu(\Omega)$ is equal to the cover strength $s(\mathscr{A},\mu)$ according to \thref{thm:CoverEfficiency}. If we choose the optimal cover protocol, then $\nu(\Omega)$ is equal to the independence degree $\gamma(G)$,
so the number of tests reduces to
\begin{align}
N&= \biggl\lceil
\frac{1}{\ln[1-2^{1-k}\gamma(G)]}\ln\delta\biggr\rceil\leq \bigl\lceil 2^{k-1}\chi(G)\ln \delta^{-1}\bigr\rceil\nonumber\\
&\leq \bigl\lceil 2^{k-1}[\Delta(G)+1]\ln \delta^{-1}\bigr\rceil\leq  \bigl\lceil 2^{k-1}n\ln \delta^{-1}\bigr\rceil, \label{eq:NumberMeasGME}
\end{align}
note that $1/\gamma(G)\leq \chi(G)\leq \Delta(G)+1\leq n$ according to \pref{pro:IndependenceDegree}.  For example, we have $N\leq 3\times 2^{k-1}\chi(G)$ when $\delta=0.05$.  For a given $\delta$, the number in \eref{eq:NumberMeasGME} is upper bounded by a constant that is independent of the number of qubits if $\gamma(G)$, $\chi(G)$, or $\Delta(G)$ is bounded. In particular,  GME of 2-colorable graph states [with $k=2$ and  $\gamma(G)=1/\chi(G)=1/2$] can be certified with only $\lceil \ln\delta/\ln(3/4)\rceil$ tests (11 tests when $\delta=0.05$);  for order-3 cluster states and Union Jack states  [with $k=3$ and  $\gamma(G)=1/\chi(G)=1/3$], it suffices to perform $\lceil \ln\delta/\ln(11/12)\rceil$  tests (35 tests when $\delta=0.05$). Incidentally,  efficient entanglement (not GME) verification of cluster states was also studied in \rcite{DimiD18}.

\subsection{Adversarial scenario}
Next, consider  certification of GME of hypergraph states in the adversarial scenario. Given  the cover protocol $(\mathscr{A},\mu)$  with verification operator $\Omega=\Omega(\mathscr{A},\mu)$, to certify the GME of $|G\>$  with significance level $\delta$, the minimal  number of required tests satisfies 
\begin{equation}\label{eq:NumberTestGMEAdv}
N\leq \biggl\lceil
\frac{2^{k-1}(1- \delta)}{\nu(\Omega)\delta}\biggr\rceil
\end{equation}
according to \eref{eq:NumTestAdvHG} with $\epsilon=2^{1-k}$,
where the spectral gap $\nu(\Omega)=s(\mathscr{A},\mu)$
depends on the specific cover protocol. For example, $\nu(\Omega)=\gamma(G)$ for the optimal cover protocol and $\nu(\Omega)=1/\chi(G)$ for the optimal coloring protocol.
When $\nu(\Omega)\geq 1/2$,  the lower bound in \eref{eq:NumTestAdvHG} is saturated,
so the number of tests reduces to
\begin{equation}\label{eq:NumberTestGMEAdv2}
N= \min\left\{
\biggl\lceil\frac{2^{k-1}(1- \delta)}{\nu(\Omega) \delta}\biggr\rceil,\; \biggl\lceil\frac{2^{k-1}}{\delta}-1\biggr\rceil\right\}.
\end{equation}

To improve the scaling of $N$ with $1/\delta$, we can apply the hedged cover protocol $(\mathscr{A},\mu)_*$  proposed in \sref{sec:HedgedCover}.
Then the number of  required tests  is given by \eref{eq:NumTestAdvHGhc}  with $\epsilon=2^{1-k}$, that is,
\begin{align}
&N=\bigl\lfloor 2^{k-1}h_*(\nu) \ln(F\delta)^{-1}\bigr\rfloor\leq  \frac{2^{k-1}\ln(F\delta)^{-1}}{\nu(1-\nu+\rme^{-1}\nu^2)}\nonumber\\
&\leq \frac{2^{k-1}(1+\rme \nu-\nu)\ln(F\delta)^{-1}}{\nu}\leq \frac{2^{k-1}\rme\ln(F\delta)^{-1}}{\nu} \label{eq:NumberTestGMEAdv3},
\end{align}
where $\nu=s(\mathscr{A},\mu)$ and $F=1-\epsilon =1-2^{1-k}$. If $\mathscr{A}$ denotes  the optimal coloring protocol, then  $\nu=1/\chi(G)$, and the number of tests required by the hedged coloring protocol $\mathscr{A}_*$ satisfies
\begin{align}\label{eq:NumberTestGMEAdv4}
N&\leq 2^{k-1}[\chi(G)+\rme-1]\ln[(1-2^{1-k})^{-1}\delta^{-1}]\nonumber\\
&\leq 2^{k-1}[\Delta(G)+\rme]\ln[(1-2^{1-k})^{-1}\delta^{-1}].
\end{align}
This equation is also applicable  if the  protocol $\mathscr{A}_*$  is replaced by  $\mathscr{A}_p$ with
$p=\nu/\rme$. These results are comparable  to the counterparts for the nonadversarial scenario presented in \eref{eq:NumberMeasGME}, especially when $k$, $\chi(G)$, and $\Delta(G)$ are large. Therefore, GME of hypergraph states can be certified efficiently even in the adversarial scenario as long as the order $k$ is bounded.
For example, GME of 2-colorable graph states  can be certified in the adversarial scenario with only $\lfloor 6.44\ln(2/\delta)\rfloor$ tests (23 tests when $\delta=0.05$) according to \eref{eq:NumberTestGMEAdv3}.  
For order-3 cluster states and Union Jack states, which are 3-colorable, it suffices to perform $\lfloor 16.3\ln(4/3\delta)\rfloor$  tests (53 tests when $\delta=0.05$).

Although detection or certification of GME has been discussed in many works, our approach is appealing for at least four reasons. First, our approach is based on quantum state verification, which can provide more precise information about the state than entanglement detection usually based on witness operators. Such information is crucial to many practical applications, including MBQC. Second, our approach requires much fewer measurement settings and tests than most previous works on the detection of GME. Third,
given a significance level, we can determine the number of required tests explicitly, which is not the case for most  previous works. Fourth, our approach can be applied to both nonadversarial scenario and adversarial scenario.

\section{\label{sec:QuditHGS}Verification of qudit hypergraph states}
Most previous verification protocols for hypergraph states only apply to the  qubit case \cite{MoriTH17,TakeM18}. Here  we show that the cover protocol and hedged cover protocol  can also be applied to qudit hypergraph states with minor modifications; in addition, most conclusions on the verification of qubit hypergraph states are still applicable in  the qudit case.
This merit is particularly appealing to both  theoretical studies and  practical applications.

\subsection{Qudit hypergraphs}
In the case of qudit, we need to revise the definition of hypergraphs to take into account multiplicities of hyperedges. Now a hypergraph  $G=(V,E,m_E)$ (also known as multihypergraph in the literature) is characterized by a set of vertices $V$ and a set of   hyperedges $E\subset \mathscr{P}(V)$ together with multiplicities specified by $m_E=(m_e)_{e\in E}$, where $m_e\in \bbZ_d$ with $m_e\neq 0$ and $\bbZ_d$ is the ring of integers modulo $d$ \cite{SteiRMG17,XionZCC18}. Nevertheless, almost all graph theoretic concepts considered in this work do not depend on the multiplicity vector $m_E$ and are defined in the same way as in the qubit case. To be specific, these concepts include
the order of a hyperedge and the hypergraph, the adjacency relation, the degree of a vertex and the hypergraph, clique and clique number, independent set and independence number,  (weighted) independence cover, cover strength, independence degree, and (fractional) chromatic number. Therefore, \pref{pro:IndependenceDegree} and its proof are applicable without any modification.

\subsection{Qudit hypergraph states}
The qudit Pauli group (also known as the Heisenberg-Weyl group) is generated by the following two generalized Pauli operators
\begin{equation}
X=\sum_{j\in \bbZ_d }|j+1\>\<j|, \quad Z=\sum_{j\in \bbZ_d } \omega^j |j\>\<j|,
\end{equation}
where $\omega=\rme^{2\pi\rmi/d}$ is a primitive $d$th root of unity.
Given any qudit hypergraph $G=(V,E,m_E)$ with $n$ vertices, we can construct an $n$-qudit hypergraph state $|G\>$ as follows: prepare the quantum state $|+\>:=\frac{1}{\sqrt{d}}\sum_{j\in \bbZ_d}|j\>$ (eigenstate of $X$ with eigenvalue~1) for each vertex of $G$ and apply $m_e$ times the generalized controlled-$Z$ operation $CZ_e$ on the vertices of each hyperedge $e$ \cite{SteiRMG17,XionZCC18},
that is,
\begin{equation}
|G\>=\Biggl(\prod_{e\in E} CZ_e^{m_e}\Biggr)|+\>^{\otimes n}.
\end{equation}
To simplify the notation, here we only give the expression of $CZ_e$ when $e=\{1,2,\ldots,k\}$, in which case we have
\begin{equation}
CZ_e:=\!\sum_{j_1, j_2,\ldots, j_k\in \bbZ_d}\omega^{j_1 j_2\ldots j_k} |j_1, j_2,\ldots, j_k\>\<j_1, j_2,\ldots, j_k|;
\end{equation}
the general case is defined analogously.

Alternatively, $|G\>$ is the unique eigenstate (up to a global phase factor) with eigenvalue 1  of the $n$ commuting (nonlocal) stabilizer operators \cite{SteiRMG17,XionZCC18}
\begin{equation}\label{eq:StabilizerQuditHGS}
K_j =X_j\otimes \prod_{e\in E|\, e\ni j	} CZ_{e\setminus \{j\}}^{m_e}, \quad j=1,2,\ldots, n.
\end{equation}
Note that  $K_j^d=\id$, so all eigenvalues of $K_j$ are powers of $\omega$.
As in the qubit case, graph theoretic concepts pertinent to the hypergraph $G$ also apply to the corresponding hypergraph state $|G\>$.

\subsection{Verification of qudit hypergraph states}
The following protocol for verifying qudit hypergraph states is a simple variation of the cover protocol for verifying qubit hypergraph states presented in \sref{sec:cover}.

Let $G=(V,E,m_E)$ be a qudit hypergraph and $|G\>$ the associated hypergraph state. Choose an independence cover  $\mathscr{A}=\{A_1, A_2, \ldots \}$  of $G$ and let $\overline{A}_l:=V\setminus A_l$ be the complement of $A_l$ in $V$. Then we can construct a verification protocol with $|\scrA|$ distinct tests (measurement settings): the $l$th test consists in measuring  $X_j$ for all  $j\in A_l$ and measuring $Z_k$ for all $k\in \overline{A}_l$. By measuring $X_j$ ($Z_k$) we mean the measurement on the eigenbasis of $X_j$ ($Z_k$).
The
measurement outcome on the $a$th qubit for $a=1,2,\ldots, n$ can be written as $\omega^{o_a}$, where $o_a\in \bbZ_d$.
Note that $X_j$ and $Z_k$ commute with $K_i$ for all $i,j\in A_l$ and $k\in \overline{A}_l$. In addition, the joint eigenstate of $X_j$ and $Z_k$ corresponding to the outcome $\{o_a\}$
is  an eigenstate of $K_i$,
whose eigenvalue is given by $\omega^{t_i}$ with
\begin{equation}
t_i=o_i+ \sum_{e\in E| e \ni i}m_e \prod_{k\in e, k\neq i}o_k
\end{equation}
according to \eref{eq:StabilizerQuditHGS}.
The test is passed if $\omega^{t_i}=1$ for all $i\in A_l$. The projector onto the pass eigenspace associated with the $l$th test reads
\begin{equation}
P_l=\prod_{i\in A_l}\biggl(\frac{1}{d}\sum_{b\in \bbZ_d} K_i^b\biggr).
\end{equation}
A state can pass all tests iff it is stabilized by  $K_i$  for all $i\in V$. So only the target state $|G\>$ can pass all tests with certainty as desired.

Suppose the $l$th test $P_l$ (associated with $A_l$) is applied with probability $\mu_l$. The efficiency of the resulting protocol is determined by the spectral gap of the verification operator $\Omega(\mathscr{A},\mu)=\sum_{l=1}\mu_l P_l$.  Here the  common eigenbasis of $K_i$ for $i\in V$ also form an eigenbasis of $\Omega(\mathscr{A},\mu)$. Each eigenstate $|\Psi_x\>$ in this basis is specified by  a string  $x\in \bbZ_d^n$ and  satisfies $K_i|\Psi_x\>=\omega^{x_i}|\Psi_x\>$.  The corresponding eigenvalue of $\Omega(\mathscr{A},\mu)$ reads
\begin{equation}
\lambda_x= \sum_{l|\supp(x)\subset \overline{A}_l} \mu_l,
\end{equation}
where $\supp(x):=\{i\,|\, x_i\neq 0\}$.
The second largest eigenvalue of $\Omega(\mathscr{A},\mu)$ can be attained when  $x_i=0$ for all $i\in V$ except for one of them. So we have
\begin{align}
\beta(\Omega(\mathscr{A},\mu))&=\max_{i\in V}\sum_{l| \overline{A}_l\ni i } \mu_l,\\
\nu(\Omega(\mathscr{A},\mu))&=\min_{i\in V}\sum_{l| A_l\ni i } \mu_l =s(\mathscr{A},\mu),
\end{align}
as in the case of qubit hypergraph states. Similarly, the smallest eigenvalue of $\Omega(\mathscr{A},\mu)$ is attained when all dits of $x$ are nonzero, in which case we have $\lambda_x=0$. Again,  the verification operator $\Omega(\mathscr{A},\mu)$ is always singular.

In addition, the hedged cover protocol can be generalized to qudit hypergraph states according to the same recipe presented in \sref{sec:HedgedCover}.
Moreover, \thref{thm:CoverEfficiency} and \esref{eq:vcoloring}-\eqref{eq:NumTestAdvHGhcolorOpt} are still applicable  in the qudit case.

\section{\label{sec:summary}Summary}  We  proposed a simple protocol---the cover protocol---for verifying (qubit and qudit) hypergraph states which  requires only two distinct Pauli measurements for each party. This protocol is dramatically more efficient than all previous protocols based on local measurements and is comparable to  the best protocol based on entangling measurements. In general,  the overhead is bounded by the chromatic number and degree of the underlying hypergraph. For many interesting hypergraph states, including Union Jack states, the number of required tests is even independent of the number of qubits.
Our protocol enables the verification of hypergraph states and  GME of thousands of qubits, which is instrumental to many applications in  quantum information processing.  Moreover, we proposed the hedged cover protocol which can be applied to verify hypergraph states and GME in the adversarial scenario
with almost the same efficiency as in the nonadversarial scenario. 
This protocol is thus  particularly appealing to many applications that require high security conditions,  such as blind MBQC and quantum networks.

\acknowledgments
We  are grateful to a referee for pointing out the relation between the independence degree and  fractional chromatic number. HZ is grateful to  Zhibo Hou and Jiangwei Shang for discussions. This work is  supported by  the National Natural Science Foundation of China (Grant No. 11875110).
 HZ acknowledges financial support from  the Excellence Initiative of the German Federal and State Governments Zukunftskonzept
 (ZUK~81) and the Deutsche Forschungsgemeinschaft (DFG) in the early stage of this work. MH was supported in part by
 Fund for the Promotion of Joint International Research (Fostering Joint International Research) Grant No. 15KK0007,
 Japan Society for the Promotion of Science (JSPS) Grant-in-Aid for Scientific Research (A) No. 17H01280, (B) No. 16KT0017, and Kayamori Foundation of Informational Science Advancement.

\appendix
\section*{Appendix}
In this appendix, we prove \pref{pro:IndependenceDegree} and provide additional details on hypergraphs, which are instructive to understanding the verification of hypergraph states. Then we  present an independent proof of   \thref{thm:GenuineEnt}, which was originally proved in \rcite{GhioMRB18}.
Finally we compare our work with previous works and demonstrate the advantage of our approach.
\section{\label{sec:HypergraphApp}Cover strengths and independence degrees of Hypergraphs}
In this section we  prove \pref{pro:IndependenceDegree} and provide additional details on independence degrees and related  invariants of hypergraphs. Here we would like to thank an anonymous referee for pointing out that the independence degree is actually equal to the inverse  fractional chromatic number  \cite{ScheU97}. In view of this fact, all conclusions presented here are well known. For the convenience of the readers, nevertheless, we give elementary proofs of the main conclusions, assuming little background on graph theory.

\subsection{Fractional coloring and fractional chromatic number}
Let $G=(V,E)$ be a hypergraph with vertex set $V$ and hyperedge set $E$.
Let $\scrI$ be the set of all independent sets of $G$.  A fractional coloring  of $G$ \cite{HiltRS73,ScheU97} is a mapping $g$ from $\scrI$ to nonnegative real numbers such that
\begin{equation}\label{eq:FractionalColorCon}
\sum_{A\in \scrI(G), \, A\ni j} g(A)\geq 1\quad \forall j\in V.
\end{equation}
The weight  $w(g)$ of the fractional coloring $g$ is defined as  the sum of $g(A)$ over all independent sets in $\scrI$, that is, $w(g)=\sum_{A\in \scrI} g(A)$. The fractional chromatic number $\chi_f(G)$ of $G$ is the minimum weight over all fractional colorings of $G$ \cite{ScheU97}. Note that there are several equivalent definitions of the fractional chromatic number; here we have chosen the definition that is the most convenient for the current study. The fractional chromatic number would reduce to the usual chromatic number if $g(A)$ for each $A\in \scrI$ could take on only two possible values 0 and~1.

Given any fractional coloring $g$ of $G$, we can construct a weighted independence cover  $(\scrI,\mu)$ of $G$ by setting $\mu(A)=w(g)^{-1} g(A)$.
Here the set $\scrI$ appearing in $(\scrI,\mu)$ can be replaced by the independence cover $\{A\in \scrI\,|\, g(A)>0\}$. 
 Conversely, given any weighted independence cover $(\scrA, \mu)$ of $G$ with nonzero cover strength $s(\scrA,\mu)$, we can construct a fractional coloring $g$ of $G$ by setting
\begin{align}
g(A)&=\frac{\mu(A)}{s(\scrA,\mu)}, \quad  A \in \scrA,
\\ g(A)&=0\quad \forall A\in (\scrI\setminus \scrA).
\end{align}
Note that $\scrA$ can also be replaced by $\scrI$ by adding zero components in $\mu$; the resulting weighted independence cover is taken to be identical to the original one. The above discussions establish a one-to-one correspondence between weighted independence covers of $G$ with nonzero cover strengths and fractional colorings $g$ which can saturate the inequality in \eref{eq:FractionalColorCon} for at least one vertex $j\in V$.  Therefore, the independence degree is equal to the inverse  fractional chromatic number \cite{ScheU97}, that is,
\begin{equation}\label{eq:InddegreeFCN}
\gamma(G)=\frac{1}{\chi_f(G)}.
\end{equation}

In view of the above observations, every fractional coloring of $G$ can be applied  to construct a cover protocol for verifying the hypergraph state $|G\>$ as presented in the main text; conversely, every cover protocol determines a fractional coloring of $G$. So a cover protocol may also be referred to as a fractional coloring protocol.

\subsection{Proof of \pref{pro:IndependenceDegree}}
Here we present an elementary and self-contained proof of \pref{pro:IndependenceDegree}. This conclusion is well known; in particular, the inequalities concerning the independence degree can be found in Chapter 3 of  \rcite{ScheU97} in view of the relation $\gamma(G)=1/\chi_f(G)$ in \eref{eq:InddegreeFCN}.
\begin{proof}
	The inequality $\frac{1}{\Delta(G)+1}\leq \frac{1}{\chi(G)}$  is equivalent to  $\chi(G)\leq \Delta(G)+1$ and follows from a well-known greedy algorithm, which produces a coloring of $G$ with no more than $\Delta(G)+1$ colors. The following algorithm is a simplified version of the Dsatur algorithm  introduced in \rcite{Brel79}. Let $v_1, v_2, \ldots, v_n$ be the vertices of $G$ whose degrees are in decreasing order. Use natural numbers to represent colors and assign color 1 to $v_1$. The colors of other vertices are assigned inductively as follows. Suppose the colors of $v_1, v_2, \ldots, v_{j-1}$ for $j\leq n$ have been assigned. Then the color number of $v_{j}$ is the smallest natural number that is different from the color numbers of those vertices in the set $\{v_1, v_2, \ldots, v_{j-1}\}$ that are adjacent to $v_{j}$. Since $v_{j}
	$ has at most $\min\{\deg(v_{j}),j-1\}$ neighbors in this set, where $\deg(v_{j})$ is the degree of $v_j$, it follows that  the color number of $j$ is at most  $\min\{\deg(v_{j})+1,j\}$. Therefore,
	\begin{equation}
	\chi(G)\leq \max_{j}  \min\{\deg(v_{j})+1,j\}\leq  \Delta(G)+1.
	\end{equation}
	
	The inequality $\gamma(G)\geq 1/\chi(G)$ follows from the observation that any  independence cover (or coloring) of $G$ with  $\chi(G)$ elements and  uniform weights has cover strength $1/\chi(G)$. Alternatively, this inequality follows from the relation $\gamma(G)=1/\chi_f(G)$ and the inequality $\chi_f(G)\leq \chi(G)$ \cite{ScheU97}.
	
	To prove the inequality $\gamma(G)\leq \alpha(G)/|V|$, let $(\mathscr{A},\mu)$ be an arbitrary independence cover. Then we have
	\begin{align}
	&|V|s(\mathscr{A},\mu)=|V|\min_{j\in V}\sum_{l|A_l \ni j} \mu_l\leq \sum_j\sum_{l|A_l \ni j}\mu_l\nonumber\\
	&=\sum_{l} \mu_l |A_l|\leq \alpha(G)\sum_l\mu_l=\alpha(G),
	\end{align}
	which implies that  $\gamma(G)\leq \alpha(G)/|V|$.

	To prove the inequality $\gamma(G)\leq 1/\varpi(G)$, let $V_{\rmC}$ be a subset of $\varpi(G)$ vertices in $V$
	that forms a clique. Then
	\begin{align}
	&\varpi(G)s(\mathscr{A},\mu)=\varpi(G)\min_{j\in V}\sum_{l|A_l \ni j} \mu_l\leq
	\sum_{j\in V_\rmC}\sum_{l|A_l \ni j}\mu_l\nonumber\\
	&\leq\sum_{l} \mu_l= 1,
	\end{align}
	where the second inequality follows from the fact that each independent set  $A_l$ can contain at most one vertex in the clique $V_\rmC$.	
\end{proof}

\subsection{\label{sec:ColorMinCover}Cover strengths of colorings and minimal covers}

Let $G=(V,E)$ be a hypergraph and
$(\scrA,\mu)$ a weighted independence cover constructed from a coloring $\scrA$, assuming that no independent set in $\scrA$ is empty (note that empty independent sets cannot increase the cover strength). Then each vertex of $V$ is contained in only one independent set in $\scrA$, which implies that
\begin{equation}\label{eq:CoverStrColoring}
s(\mathscr{A},\mu)=\min_{l} \mu_l\leq |\scrA|^{-1}\leq \chi(G)^{-1}.
\end{equation}
Here the first inequality is saturated iff all weights $\mu_l$ are equal, and the second inequality is saturated iff the coloring $\scrA$ is optimal in the sense that no other coloring of $G$ requires fewer colors.

Next, let $(\scrA,\mu)$ be a weighted independence cover of $G$ constructed from a minimal cover $\scrA$. By "minimal" we mean that any proper subset $\scrA'$ of $\scrA$ is not a cover of $G$ because the union of independent sets in $\scrA'$ does not coincide with the vertex set $V$. In other words, for any $A_l$ in $\scrA$, there exists a vertex $j\in V$ such that $j\in A_l$ and $j\notin A_k$ for all $k\neq l$. Therefore,  \begin{equation}\label{eq:CoverStrMinimalC}
s(\mathscr{A},\mu)=\min_{l} \mu_l\leq |\scrA|^{-1}\leq \chi(G)^{-1}
\end{equation}
as in \eref{eq:CoverStrColoring}. Again the first inequality is saturated iff all weights $\mu_l$ are equal;  the second inequality is saturated iff $|\scrA|=\chi(G)$, in which case an optimal coloring of $G$ can be constructed from $\scrA$ by deleting some vertices in some independent sets if $\scrA$ is not yet a coloring.

In view of the above discussion, to maximize the cover strength it is always beneficial to choose uniform weights when $\scrA$ is a coloring or  minimal independence cover.
In addition, the cover strength of any such  cover is upper bounded by $1/\chi(G)$, which can be saturated.

\subsection{Independence degrees of odd cycles}
Let $C_n$ be a  cycle with $n$ vertices, where $n\geq 3$ is an odd integer. According to Proposition 3.1.2 in \rcite{ScheU97}, the fractional chromatic number of $C_n$ reads
\begin{equation}
\chi_f(C_n)=\frac{2n}{n-1}. 
\end{equation}
Thanks to \eref{eq:InddegreeFCN},
the independence degree of $C_n$   is thus given by
\begin{equation}\label{eq:IndDegreeOddC}
\gamma(C_n)=\frac{n-1}{2n}=\frac{1}{2}-\frac{1}{2n},
\end{equation}
which increases monotonically with $n$.
This conclusion  indicates  that overcomplete covers of some hypergraph $G$ can have cover strengths larger than $1/\chi(G)$ and that the inequality $\gamma(G)\geq 1/\chi(G)$ in \pref{pro:IndependenceDegree} is in general strict. By contrast, the cover strength of any  coloring or minimal cover  of $C_n$ is upper bounded by $1/3$ given that $\chi(C_n)=3$. So it is indeed advantageous to consider weighted independence covers (or fractional colorings) beyond usual colorings for some hypergraphs. These observations are of interest to constructing  efficient verification protocols for hypergraph states (including graph states in particular) in view of \thref{thm:CoverEfficiency} in the main text.

To prove \eref{eq:IndDegreeOddC}, note that $\alpha(C_n)=(n-1)/2$, so that  $\gamma(C_n)\leq (n-1)/(2n)$ according to \pref{pro:IndependenceDegree}. This upper bound can be saturated by the equal-weight cover composed of the $n$ independent sets
\begin{equation}\label{eq:OddCycleCover}
A_j=\{j, j+2, \ldots, j+n-3\},\quad  j=1,2,\ldots, n.
\end{equation}
Here vertex labels $j$ and $j+n$ are identified.

\section{\label{sec:GMEHGSproof}Proof of \thref{thm:GenuineEnt}}
In this section we present an independent proof of \thref{thm:GenuineEnt}, which was originally proved in \rcite{GhioMRB18}.
This theorem is an immediate consequence of \lsref{lem:GenuineEnt} and~\ref{lem:kappaG} presented below. Before stating and proving these  auxiliary results, we need to introduce a few additional concepts. Let $|\Psi\>\<\Psi|$ be an $n$-partite pure state of the parties $V=\{1,2,\ldots, n\}$. For each nonempty proper subset $A$ of $V$, denote by $\varrho_A$ the reduced state of $|\Psi\>\<\Psi|$ over the parties in $A$, that is,  $\varrho_A=\tr_{\overline{A}}(|\Psi\>\<\Psi|)$, where $\overline{A}=V\setminus A$ is the complement of $A$ in $V$. Define
\begin{equation}
\kappa(|\Psi\>):=\max_{A}\|\varrho_{A}\|,
\end{equation}
where $\|\varrho_{A}\|$ denotes the operator norm (the largest eigenvalue) of $\varrho_A$  and the maximum is taken over all nonempty proper subsets $A$ of $V$. Note that $\kappa(|\Psi\>)$ is invariant under local unitary transformations. Given a hypergraph $G$, define $\kappa(G):=\kappa(|G\>)$.

\Lsref{lem:GenuineEntPure} and \ref{lem:GenuineEnt} below were known previously \cite{GuhnT09}, but here we provide self-contained proofs for completeness.
\begin{lemma}\label{lem:GenuineEntPure}
	The pure state $|\Psi\>$ is GME iff $\kappa(|\Psi\>)<1$.
\end{lemma}
\begin{proof}
	To prove the lemma, it is equivalent to prove the statement that the state $|\Psi\>$ is biseparable iff $\kappa(|\Psi\>)=1$. If  $\kappa(|\Psi\>)=1$, then $|\Psi\>$ has a nontrivial reduced state that is pure, which implies that $|\Psi\>$ is biseparable. Conversely, if $|\Psi\>$ is biseparable, then it has a nontrivial reduced state that is pure, which implies that  $\kappa(|\Psi\>)=1$.
\end{proof}

\begin{lemma}\label{lem:GenuineEnt}
	Suppose $|\Psi\>$ is GME and a state $\rho$  satisfies  $\<\Psi|\rho|\Psi\>> \kappa(|\Psi\>)$. Then $\rho$ is GME.
\end{lemma}
\begin{proof}
	Suppose  $|\Phi\>$ is an arbitrary pure   state that is biseparable over the partition $A$ and $\overline{A}$, that is, $|\Phi\>$ has the form  $|\Phi\>=|\Phi_A\>\otimes |\Phi_{\overline{A}}\>$. Then
	\begin{equation} |\<\Psi|\Phi\>|^2\leq \<\Phi_A|\varrho_A|\Phi_A\>\leq \|\varrho_A\|
	\leq \kappa(|\Psi\>),
	\end{equation}
	where $\varrho_A$ is the reduced state of $|\Psi\>\<\Psi|$ over the parties in $A$.
	If $\rho$ is not GME, then it is a convex combination of biseparable pure states, so that $\<\Psi|\rho|\Psi\>\leq \kappa(|\Psi\>)$. Therefore, $\rho$ is GME whenever $\<\Psi|\rho|\Psi\>> \kappa(|\Psi\>)$.
\end{proof}

\begin{lemma}\label{lem:kappaG}
	Suppose $G=(V,E)$ is a connected order-$k$ hypergraph with $k\geq 2$. Then $\kappa(G)\leq 1-2^{1-k}$.
\end{lemma}
\begin{proof}
	This lemma is an easy consequence of \lref{eq:ReduceStateNorm} below.  When $|G\>$ is a connected graph state,
	\lref{lem:kappaG} is known much earlier \cite{TothG05, TothG05E,GuhnT09}, in which case the bound $\kappa(G)\leq 1-2^{1-k}$ with $k=2$ is always saturated. This conclusion follows from the fact that any nontrivial reduced density matrix of the graph state $|G\>\<G|$ is proportional to a projector of rank at least 2.
\end{proof}
Besides the application in proving \thref{thm:GenuineEnt}, \lref{lem:kappaG} shows that any order-$k$ hypergraph state $|G\>$ with $k\geq 2$ and $\kappa(G)= 1-2^{1-k}$ is not equivalent to any order-$k'$ hypergraph state with $k'<k$ under local unitary transformations.

The bound $\kappa(G)\leq 1-2^{1-k}$ is saturated if $G$ contains an order-$k$ leaf. Here a leaf of $G$ is a vertex that belongs to only one hyperedge with order at least 2. The order of the leaf is the order of this unique hyperedge.
In this case $\|\varrho_A\|=1-2^{1-k}$ when $A$ is composed of the leaf.
To verify this claim, it suffices to consider the scenario in which $n=|V|=k$ and $G$ contains a single hyperedge (which necessarily has order $k$). Now it is straightforward to verify that each single-qubit reduced state of $|G\>$ has two eigenvalues equal to $1-2^{1-k}$ and $2^{1-k}$, respectively, so the bound $\kappa(G)\leq 1-2^{1-k}$ is indeed saturated. In particular, the above observation implies that $\kappa(G)= 1-2^{1-k}$ when $|G\>$ is a 1D order-$k$ cluster state. Straightforward calculations also show that the bound $\kappa(G)\leq 1-2^{1-k}$ with $k=3$ is saturated for 2D order-3 cluster states and Union Jack states for which $\kappa(G)=3/4$.

\begin{lemma}\label{eq:ReduceStateNorm}
	Suppose $G=(V,E)$ is a hypergraph and $A$ is any nonempty proper subset of $V$ that is adjacent to $\overline{A}$. Let $\varrho_A$ be the reduced state of $|G\>$ over the parties in $A$.
	Then $\|\varrho_A\|\leq 1-2^{1-k}$, where $k$ is the maximum order of hyperedges that connect $A$ and $\overline{A}$.
\end{lemma}
Here two disjoint nonempty subsets $A$ and $B$ of the vertex set $V$ of $G$ are adjacent if $E$ contains a hyperedge that connects a vertex in $A$ and a vertex in $B$.

\begin{proof}
	\Lref{eq:ReduceStateNorm} can be proved by induction. To start with,
	let  $n=|V|$; then $n\geq k\geq 2$  by assumption.	It is straightforward to verify that the lemma holds when $n=k=2$. Suppose the lemma holds for $2\leq k\leq n\leq n_0$ with $n_0\geq 2$. We shall prove that the lemma also holds for  $2\leq k\leq n= n_0+1$.

	It is instructive to note that $\|\varrho_A\|$ does not change if we add or delete hyperedges among vertices in $A$ or  hyperedges among vertices in $\overline{A}$. So we may assume that $G$ has neither hyperedges among  vertices in $A$  nor hyperedges among vertices in $\overline{A}$; in other words, every hyperedge of $G$ contains at least one vertex in $A$ and one vertex in $\overline{A}$.
	Then $k$ is equal to the order of $G$. In addition,  we may assume that $G$ has no isolated vertices. Note that the order of $G$ does not change if any isolated vertex, say $j$, is deleted; meanwhile, $\varrho_A$ does not change after this deletion  if $j\in \overline{A}$, while $\|\varrho_A\|=\|\varrho_{A\setminus \{j\}}\|$ if $j\in A$. Furthermore, we may  assume $|A|\leq n-2$ without loss of generality given that $\|\varrho_A\|=\|\varrho_{\overline{A}}\|$ and $n\geq 3$. By relabeling the parties if necessary, we may assume that $n\notin A$,  that is, $n\in \overline{A}$.
	
	According to Proposition 7.16 of \rcite{LyonUWY15},
	\begin{equation}
	\varrho_{V\setminus \{n\}}=\frac{1}{2}(|G_0\>\<G_0|+ |G_1\>\<G_1|),
	\end{equation}
	where $G_0,G_1$ are subhypergraphs of $G$ defined as follows
	\begin{equation}\label{eq:subhypergraph01}
	\begin{aligned}
	G_0&=(V\setminus\{n\},\{e\in E \,|\, n \notin e \}),\\
	G_1&=(V\setminus\{n\},\{e\in E\,|\, n \notin e \}\Delta \{e\setminus \{n\}\,|\, n\in e\in E \}).
	\end{aligned}
	\end{equation}
	Here $A \Delta B$ denotes the symmetric difference of $A$ and $B$, that is, $(A\cup B) \setminus( A\cap B)$. Literally, $G_0$ is the subhypergraph of $G$ obtained by deleting the vertex $n$ and all the hyperedges containing $n$; $G_1$ is the subhypergraph of $G$ obtained by deleting the vertex $n$, shrinking all the hyperedges containing $n$, and then deleting repeated hyperedges.
	
	Let $B=V\setminus \{n\}\setminus A$; note that $B$ is nonempty due to the assumption $|A|\leq n-2$.   In addition,
	$A\cup B=V\setminus\{n\}$ is the vertex set of both $G_0$ and $G_1$. Let
	$\varrho_0=\tr_B(|G_0\>\<G_0|)$ and $\varrho_1=\tr_B(|G_1\>\<G_1|)$. Then $\varrho_A=(\varrho_0+\varrho_1)/2$ and
	\begin{equation}
	\|\varrho_A\|\leq \frac{1}{2}(\|\varrho_0\|+\|\varrho_1\|).
	\end{equation}
	If $A$ is adjacent to $B$ with respect to both $G_0$ and $G_1$, then  the induction hypothesis implies that
	\begin{equation}
	\|\varrho_0\|\leq 1-2^{1-k_0}\leq 1-2^{1-k},\;\;  \|\varrho_1\|\leq 1-2^{1-k_1}\leq  1-2^{1-k},
	\end{equation}
	which in turn implies that $\|\varrho_A\|\leq  1-2^{1-k}$. Here   $k_0$ and $k_1$ denote the orders of $G_0$ and $G_1$, respectively, which satisfy  $k_0,k_1\leq k$.
	
	If $A$ is not adjacent to $B$ with respect to $G_0$, then $G_0$ has no  hyperedges, which implies that all hyperedges of $G$ contain the vertex $n$. Recall that  by assumption $G$ has neither hyperedges among  vertices in $A$  nor hyperedges among vertices in $\overline{A}$. Consequently, $G_1$ has order at most $k-1$. If, in addition,  $A$ is adjacent to $B$ with respect to  $G_1$, then $\|\varrho_1\|\leq 1-2^{2-k}$, which implies that
	\begin{equation}
	\|\varrho_A\|\leq \frac{1}{2}(\|\varrho_0\|+\|\varrho_1\|)\leq \frac{1}{2}(1+1-2^{2-k})=1-2^{1-k}.
	\end{equation}
	Otherwise, if $A$ is not adjacent to $B$ with respect to  $G_1$, then no hyperedge of $G$ contains any vertex in $B$; in other words, all vertices of $B$ are isolated with respect to $G$, which contradicts  our assumption.
	
	It remains to consider the case in which $A$ is adjacent to $B$ with respect to $G_0$, but not adjacent to $B$ with  respect to  $G_1$. In view of \eref{eq:subhypergraph01}, we conclude that $G_0$ has order at most $k-1$ since, otherwise, any order-$k$ hyperedge of $G_0$ (which necessarily connects $A$ and $B$) would also be a hyperedge of $G_1$. Therefore,
	\begin{equation}
	\|\varrho_A\|\leq \frac{1}{2}(\|\varrho_0\|+\|\varrho_1\|)\leq \frac{1}{2}(1-2^{2-k}+1)=1-2^{1-k}.
	\end{equation}	
	This observation completes the proof of \lref{eq:ReduceStateNorm}.
\end{proof}

\section{\label{sec:Comparison}Comparison with previous works}
In this section we discuss the connections and distinctions between our work and entanglement detection. We then compare our protocols for verifying hypergraph states with a number of previous works, including direct fidelity estimation (DFE) \cite{FlamL11} and
\rscite{MoriTH17,PallLM18,HayaH18,TakeM18,TakeMMM19}.

\subsection{Quantum state verification and entanglement detection}
In the main text, we introduce a simple and efficient protocol for verifying general hypergraph states. Our protocol can also be applied to detecting GME, though it is not necessarily optimized for this purpose.
In the literature, there are many works on the  detection of entanglement, including GME in particular \cite{GuhnT09}. The main distinction between state verification and entanglement detection lies in the motivations, which are reflected in  the following two questions.
\begin{enumerate}
	\item Is the quantum state prepared good enough for a given task, such as quantum computation, quantum communication, or quantum metrology?
	
	\item Is the quantum state prepared GME?
\end{enumerate}
The main motivation of the current work is to provide an efficient  tool for answering the first question, while most works on entanglement detection focus on the second question directly. Question 2 is definitely interesting in itself since GME is a key resource in quantum information processing and a focus of foundational studies. In addition, demonstrating GME in experiments is usually highly nontrivial and may serve as a signature of the advance in quantum information science. On the other hand,
although there are intimate connections between the two questions, the answer to question 2 is, in general, far from enough for  answering question 1, which usually entails high fidelity with the target state. Instead of demonstrating certain quantum signature, eventually, we need to answer more specific and practical questions directly. Crucial to achieving this task is efficient quantum  state verification, which is the focus of this work.

In addition, most works on entanglement detection are based on the expectation values of certain witness operators and usually do not discuss the number of tests required to make a conclusion. With the cover protocol, by contrast, we can not only provide more precise information about the quantum state prepared, but also determine the explicit number of tests required.  In addition, our approach can be applied to the adversarial scenario, which is appealing to many applications that require high security conditions.

\subsection{\label{sec:DFE}Comparison with direct fidelity estimation \rcite{FlamL11}}

In this section we compare our cover protocol with DFE introduced by Flammia and Liu \cite{FlamL11}. Compared with the cover protocol, DFE  can be applied to any pure state and thus has wider applications. The number of measurements required by DFE is smaller than tomography by a factor of $D=2^n$, where $n$ is the number of qubits. Moreover, this number does not increase with the number of qubits in the case of  stabilizer states. From this perspective, DFE is very efficient and very useful.
However, DFE has several
drawbacks as mentioned below which limit its applications to hypergraph states and many other states of
quantum systems of more than 15 qubits.

\begin{enumerate}
	\item To apply DFE it is necessary to  sample from the squared characteristic function defined on the discrete phase space of $2^{2n}$ points. In general, it is not easy to compute and store this function for large quantum systems; also, it is not easy  to implement the sampling even if the characteristic function is determined.
	
	\item The number of potential measurement settings increases exponentially with the number of qubits even for stabilizer states. The number of actual
	measurement settings $\lceil 1/(\epsilon^2 \delta)\rceil$ depends on the target infidelity $\epsilon$ and significance level $\delta$.
	Specific measurement settings cannot be determined before implementing the protocol. Also, the total number of  measurements cannot be determined in advance.
	
	\item The average total number of  measurements reads
	\begin{align}
	N_{\dfe}&\approx 1+\frac{1}{\epsilon^2\delta}+\frac{2g}{D\epsilon^2}\ln (2/\delta)\nonumber\\
	&=1+\frac{1}{\epsilon^2\delta}+\frac{2\tilde{g}}{\epsilon^2}\ln (2/\delta),\label{eq:DFEnumMeas}
	\end{align}
	where $D=2^n$, $\tilde{g}=g/2^n$, and  $g$ is the number of points  at which the characteristic function is nonzero \cite{FlamL11}. It is known that $g\geq D$ and the lower bound is saturated iff the target state is a stabilizer state.  For a generic state $g\approx D^2$, so the number of measurements increases exponentially with $n$. As we shall see shortly, the exponential growth is also inevitable for many hypergraph states.

	The number $N_{\dfe}$ in \eref{eq:DFEnumMeas} can be reduced for a well-conditioned state $\rho$,  which means either  $|\tr(\rho W_{x,z})|=0$ or $|\tr(\rho W_{x,z})|\geq c $ for all Pauli operators $W_{x,z}$ [cf.~\eref{eq:PauliW} below], where $c$ is a positive constant whose inverse is upper bounded by a polynomial of $n$. In this case, $N_{\dfe}$ can be reduced to $O\bigl(\ln(1/\delta)/(c^2\epsilon^2)\bigr)$, though the quadratic scaling behavior with $1/\epsilon$ does not change. However, many hypergraph states are not well conditioned.
	In addition, no simple way is known for determining whether a generic hypergraph state is well conditioned or not when the number of qubits is large.
\end{enumerate}

To analyze the supports of the characteristic functions of hypergraph states, it is instructive to point out that  any hypergraph state is a real equally weighted state and vice versa \cite{QuWLB13,RossHBM13}. In other words, any $n$-qubit hypergraph state can be written as
\begin{equation}
|\Psi_f\>=2^{-n/2}\sum_{x=0}^{2^n-1}(-1)^{f(x)}|x\>,
\end{equation}
where $x$ is understood as a string in $\{0,1\}^n$ and
$f$ is a Boolean function from $\{0,1\}^n$ to $\{0,1\}$. For example, the Boolean function corresponding to the hypergraph state $|G\>=\bigl(\prod_{e\in E} CZ_e\bigr)|+\>^{\otimes n}$ is given by
\begin{equation}
f(x)=\sum_{e\in E}\prod_{j\in e} x_j,
\end{equation}
where the addition is modulo 2.
Up to  a phase factor, any $n$-qubit Pauli operator can be written as
\begin{equation}\label{eq:PauliW}
W_{x,z}: =\left(\prod_{j=1}^n X_j^{x_j}\right)\left(\prod_{j=1}^n Z_j^{z_j}\right),\quad x,z\in \{0,1\}^n,
\end{equation}
where $X_j$ and $Z_j$ are the Pauli $X$ and $Z$ operators for the $j$th qubit. Here we are mainly interested in the absolute value of the characteristic function, so the overall phase factor does not matter. Calculation shows that
\begin{equation}
\<\Psi_f | W_{x,z}|\Psi_f\>=\frac{1}{2^n}\sum_{u=0}^{2^n-1}(-1)^{f(u)+f(u+x)}(-1)^{z\cdot u},
\end{equation}
where the addition $u+x$ is modulo 2 and so is the product $z\cdot u:=\sum_{j=1}^n z_j u_j$. The cardinality of the support of the characteristic function reads
\begin{equation}
g(f)= \bigl|\bigl\{(x,z)\in \{0,1\}^{2n}\,|\, \<\Psi_f | W_{x,z}|\Psi_f\>\neq0\bigr\}\bigr|.
\end{equation}

In the rest of this section, we provide several concrete examples of hypergraph states for which $\tilde{g}=g/2^n$ increases exponentially with the number $n$ of qubits, which means $N_{\dfe}$ increases exponentially with $n$.
First, consider the special hypergraph with only one hyperedge, which contains all $n$ vertices. The corresponding Boolean function $f_n$ reads
\begin{equation}
f_n(u):=\prod_{j=1}^n u_j=\begin{cases}
1 &u=11\cdots 1,\\
0 &\mbox{otherwise}.
\end{cases}
\end{equation}
In this case, we have
\begin{widetext}
	\begin{equation}
	2^n\bigl|\<\Psi_{f_n} | W_{x,z}|\Psi_{f_n}\>\bigr|=\begin{cases}
	2^n &x=z=0,\\
	2^n-4 &z=0, x\neq 0, \\
	4 & x\neq 0,z\neq 0, x\cdot z=0,\\
	0 & x\cdot z=1, \mbox { or }x=0, z\neq 0,
	\end{cases}
	\end{equation}
\end{widetext}
which implies that
\begin{equation}
g(f_n)=2^{2n-1}-2^{n-1}+1,\quad \tilde{g}\approx 2^{n-1}-2^{-1}.
\end{equation}
So the number of measurements in \eref{eq:DFEnumMeas} reduces to
\begin{equation}
N_{\dfe}\approx 1+\frac{1}{\epsilon^2\delta}+\frac{2^n-1}{\epsilon^2}\ln (2/\delta),
\end{equation}
which increases exponentially with the number of qubits. By contrast, the number of tests required by our cover protocol is at most $\lceil(n/\epsilon)\ln (1/\delta)\rceil$
by \eref{eq:NumTestCoverUB} in the main text, which is exponentially smaller than $N_{\dfe}$.

As another example, consider the tensor power  $|\Psi_{f_3}\>^{\otimes n/3}$, which corresponds to the hypergraph state with $n/3$ disjoint hyperedges of order $3$, assuming that $n$ is divisible by 3. In this case,
\begin{equation}
g=g(f_3)^{n/3}=29^{n/3}>3^n,\quad \tilde{g}=\frac{29^{n/3}}{2^n}
> \Bigl(\frac{3}{2}\Bigr)^n.
\end{equation}
So the number of measurements in \eref{eq:DFEnumMeas} reduces to
\begin{align}
N_{\dfe}&\approx 1+\frac{1}{\epsilon^2\delta}+\frac{2\times 29^{n/3}}{2^n\epsilon^2}\ln (2/\delta)\nonumber\\
&> 1+\frac{1}{\epsilon^2\delta}+\frac{2\left(\frac{3}{2}\right)^{n}}{\epsilon^2}\ln (2/\delta),
\end{align}
which also increases exponentially with the number of qubits.
By contrast, the number of tests required by the cover protocol is  $\lceil(3/\epsilon)\ln (1/\delta)\rceil$,
which is again  exponentially smaller than $N_{\dfe}$.

Furthermore, numerical calculations  show that $\tilde{g}$ increases exponentially with $n$ for  order-3 cluster states and Union Jack states on a chain or  a two-dimensional  lattice, so $N_\dfe$ also increases exponentially with $n$ for these states (cf.~\fref{fig:NumberMeas} in the main text). The number of tests required by the cover protocol is still $\lceil(3/\epsilon)\ln (1/\delta)\rceil$.

\subsection{\label{sec:MTH}Comparison with \rcite{MoriTH17}}
Recently, Morimae, Takeuchi, and Hayashi (MTH) \cite{MoriTH17} introduced a method for verifying hypergraph states in the adversarial scenario. They only considered the case in which all hyperedges have orders at most 3. Although their method may potentially be extended to more general settings,
a direct extension of their approach entails exponential increase in the resource overhead  with the order of the hypergraph. Even for order-3 hypergraph states, the resource overhead  increases exponentially  with the number of hyperedges (and thus the degree of the hypergraph). Another drawback of the MTH protocol is that even the target hypergraph state $|G\>$ cannot pass the test with certainty. Consequently, the number of tests required increases quadratically with the inverse infidelity.

More specifically, suppose $|G\>$ is an $n$-qubit hypergraph state to be verified. Let $K_j$ be the stabilizer operator corresponding to vertex $j$ as defined in \eref{eq:StabilizerHGS} in the main text; let $r_j$ be the number of order-3 hyperedges that contain the vertex $j$. The MTH verification protocol is composed of $n$ stabilizer tests. For each stabilizer $K_j$, MTH devised a test, which requires
$4^{r_j}$ potential measurement settings.
The total number  of potential measurement settings  is  $\sum_{j=1}^n 4^{r_j}$,
which increases exponentially with the number of order-3 hyperedges. MTH also set  a criterion such that the probability of a state $\rho$ to satisfy the criterion is given by
\begin{equation}
p_j=\frac{1}{2} +\frac{\tr(\rho K_j)}{2^{r_j+1}}=\frac{1}{2} +\frac{1-a_j}{2^{r_j+1}},
\end{equation}
where $a_j:=1-\tr(\rho K_j)$ satisfies $0\leq a_j\leq 2$.
Although the target state $|G\>$ can attain the maximum probability $(1/2)+(1/2^{r_j+1})$,  it generally cannot satisfy the criterion with certainty. Suppose the test is performed $N_j$ times,
and the criterion is satisfied $t_j$ times. Then the stabilizer test is passed if the frequency $f_j=t_j/N_j$ satisfies
\begin{equation}
f_j\geq \frac{1}{2}+\frac{1-\theta}{2^{r_j+1}},
\end{equation}
where $\theta$ is a small positive constant.
The state $\rho$ is accepted if it can pass all the stabilizer tests.
The choice of $\theta$ needs to guarantee that the target state $|G\>$ can pass all  the tests with high probability; meanwhile, any state that has low fidelity with $|G\>$ should fail some test with high probability.
When  $a_j\geq \theta$,
the probability that $\rho$ can pass the stabilizer test associated with $K_j$ can be upper bounded as follows,
\begin{align}
&\Pr\left(f_j\geq \frac{1}{2}+\frac{1-\theta}{2^{r_j+1}} \right)=\Pr\left(f_j\geq p_j+\frac{a_j-\theta}{2^{r_j+1}} \right)\nonumber\\
&\leq \exp\left(-2\frac{(a_j-\theta)^2}{4^{r_j+1}}N_j\right), \label{eq:Hoeffding}
\end{align}
where the last step follows from the Hoeffding  inequality. Similarly, the probability that the target state $|G\>$ passes the test can be lower bounded by virtue of  the Hoeffding  inequality.

MTH did not give an explicit number of tests needed to verify a hypergraph state within infidelity $\epsilon$ and significance level $\delta$. They considered a related, but different verification problem with a different criterion, which requires about $nk+m$ tests, where $k=2^{2r+3}n^7$, $m\geq (2\ln 2)n^7 k^2$, and $r=\max_j r_j$.
In other words, the number of required tests satisfies
\begin{align}
nk+m &\ge nk+ (2\ln 2)n^7 k^2 \cong(2\ln 2)n^7 k^2 \nonumber\\
&= (2^{4r+7}\ln 2) n^{21}. \label{eq:MTHbound1}
\end{align}
While this number is still polynomial in $n$ if $r$ does not increase with $n$, it grows rapidly with $n$. Actually, it is already astronomical when  $n=5$ and $r=2$ (note that $r=8$ for generic Union Jack states on 2D lattices), while any useful MBQC would require more than five qubits.
So the MTH protocol  is hardly practical. In contrast, the number of tests required  by our cover or coloring protocol satisfies
\begin{align}\label{eq:NumComMTH}
N\leq \biggl\lceil\frac{\Delta(G)+1}{\delta\epsilon} \biggr\rceil\leq \biggl\lceil\frac{2r+1}{\delta\epsilon} \biggr\rceil
\end{align}
according to \eref{eq:NumTestAdvHG2}. Note that $\Delta(G)\leq 2r$ since $G$ is an order-3 hypergraph state. The upper bound in \eref{eq:NumComMTH} 
is independent of $n$,  which shows that our protocol is dramatically more efficient than the MTH protocol. According to \esref{eq:NumTestAdvHGhcolor} and \eqref{eq:NumTestAdvHGhcolorOpt}, the hedged cover or coloring protocol can further reduce the number of tests to
\begin{align} N\leq \!\!\frac{[\Delta(G)+\rme]\ln(F\delta)^{-1}}{\epsilon}
\leq \frac{(2r+\rme)\ln(F\delta)^{-1}}{\epsilon}.
\end{align}

It is natural to ask whether the number of tests
can be reduced significantly if the MTH protocol is adapted to the nonadversarial
scenario considered in the main text. Here we try to give a rough estimate.

To verify $|G\>$ within infidelity $\epsilon$ and significance level $\delta$, suppose $1-\<G|\rho|G\>\geq \epsilon$, we need to estimate the maximal probability that $\rho$ can pass all the stabilizer tests and make sure that this probability is smaller than $\delta$, that is,
\begin{equation}
\prod_j \Pr\left(f_j\geq \frac{1}{2}+\frac{1-\theta}{2^{r_j+1}} \right)=\Pr\left(f_j\geq p_j+\frac{a_j-\theta}{2^{r_j+1}} \right)\leq \delta.
\end{equation}
According to \eref{eq:Hoeffding}, it suffices to guarantee that
\begin{equation}\label{eq:ProbBound}
\prod_{j|a_j\geq \theta }          \exp\left(-2\frac{(a_j-\theta)^2}{4^{r_j+1}}N_j\right)\leq \delta.
\end{equation}
Note that  the infidelity of $\rho$ with $|G\>$ satisfies
\begin{align}
&1-\<G|\rho|G\> =1-\tr\Biggl(\rho\prod_j \frac{K_j+1}{2}\Biggr)\nonumber\\
&\leq \sum_j \biggl[1-\tr\biggl(\rho \frac{K_j+1}{2}\biggr)\biggr]=\frac{1}{2} \sum_j a_j.\label{eq:UnionBound}
\end{align}
The upper bound can be saturated when $0\leq \sum_j a_j\leq 2$. To see this, for each $k=1,2,\ldots, n$, let $|G_k\>$ be the common eigenstate of $K_j$ for $j=1,2,\ldots, n$ with eigenvalues
\begin{equation}
K_k|G_k\>=-1, \quad K_j|G_k\> =1, \quad \forall j\neq k. 
\end{equation}
Let 
\begin{equation}
\rho=\biggl(1-\sum_j \lambda_j \biggr) |G\>\<G|+\sum_j \lambda_j |G_j\>\<G_j|,
\end{equation}
where $\lambda_j=a_j/2$. Then we have $1-\tr(\rho K_j)=2\lambda_j=a_j$ and
$1-\<G|\rho|G\>=\sum_j\lambda_j=\sum_j a_j/2$, so the bound in \eref{eq:UnionBound}  is saturated.

According to \eref{eq:UnionBound},  $\sum_j a_j\geq 2\epsilon$ if
the infidelity of $\rho$ is at least $\epsilon$, that is, $1-\<G|\rho|G\geq \epsilon$.
Now we can  derive a lower bound for $\sum_j N_j$ under the requirement that
\eref{eq:ProbBound} holds whenever
$\sum_j a_j\geq 2\epsilon$. To this end, choose
\begin{equation}
a_j=\frac{2\epsilon\times 2^{r_j}}{\sum_k 2^{r_k}},
\end{equation}
then we have $\sum_j a_j=2\epsilon$. So 
\eref{eq:ProbBound} implies that
\begin{equation}
\exp\left(-\frac{2\epsilon^2\sum_j N_j}{\bigl(\sum_j 2^{r_j}\bigr)^2}\right)\leq \delta,
\end{equation}
which in turn implies that
\begin{equation}\label{eq:MTHbound2}
N_{\mth}=\sum_j N_j\geq \frac{\bigl(\sum_j 2^{r_j}\bigr)^2\ln\delta^{-1}}{2\epsilon^2}.
\end{equation}
If all $r_j$ are equal to $r$, then the MTH protocol requires $4^r n$ potential measurement settings  and at least
\begin{equation}\label{eq:MTHbound3}
N_{\mth}\geq \frac{4^r n^2\ln\delta^{-1}}{2\epsilon^2}
\end{equation}
tests. The bounds in the above two equations have much  better scaling behavior with $n$ compared with the bound in \eref{eq:MTHbound1}. However, these bounds are already very large for a small value of $n$ for Union Jack states and many other hypergraph states for which $r$ is not so small.
In general, it is too prohibitive to implement the MTH protocol
except for hypergraph states of no more than ten qubits.

A few comments are in order. First, we do not know how tight the bounds in \esref{eq:MTHbound2} and \eqref{eq:MTHbound3} are. Nevertheless, these bounds are sufficient for comparing the MTH protocol with our protocol, and it is  not so important to derive a tighter bound using more involved analysis. Second,  \eref{eq:MTHbound2} is based on
\esref{eq:Hoeffding} and \eqref{eq:UnionBound}. Note that the bound in \eqref{eq:UnionBound} is tight. The Hoeffding inequality in  \eref{eq:Hoeffding} may potentially be improved, thereby reducing $N_\mth$. However, this possibility was not considered by MTH. We are not aware of any simple method for improving the Hoeffding inequality either and do not expect a significant improvement even with more sophisticated analysis.  In this regard, our protocol is not only much more efficient, but also much easier to implement and to analyze its performance.

In the rest of this section, we consider the performance of the MTH protocol adapted to the nonadversarial scenario for several concrete  order-3 hypergraph states.
As a start, consider the complete order-3 hypergraph state whose underlying hypergraph contains all possible order-3 hyperedges. In this case, the total number of hyperedges is $\binom{n}{3}=n(n-1)(n-2)/6$ and $r_j=r=\binom{n-1}{2}=(n-1)(n-2)/2$ for $j=1,2,\ldots, n$. Therefore,
\begin{equation}\label{eq:MTHboundComplete}
N_{\mth}\geq \frac{2^{(n-1)(n-2)} n^2\ln\delta^{-1}}{2\epsilon^2}.
\end{equation}
Here both the number of potential measurement settings and the number of tests required by the MTH protocol increase exponentially with the number of qubits. By contrast, our cover protocol requires at most $n$ potential measurement settings and $\lceil(n/\epsilon)\ln(1/\delta)\rceil$ tests according to \eref{eq:NumTestCoverUB}.

The  examples considered in the rest of this section are 3-colorable, so our cover protocol requires three measurement settings and $\lceil(3/\epsilon)\ln (1/\delta)\rceil$ tests to verify each hypergraph state within infidelity $\epsilon$ and significance level~$\delta$. First, consider the tensor power  $|\Psi_{f_3}\>^{\otimes n/3}$ introduced in \aref{sec:DFE}, assuming that $n$ is divisible by 3. In this case $r_j=r=1$ for all $j=1,2,\ldots,n$. So \eref{eq:MTHbound3} reduces to
\begin{equation}\label{eq:MTHboundf3Power}
N_{\mth}\geq \frac{2 n^2\ln\delta^{-1}}{\epsilon^2}.
\end{equation}

Next, consider  order-3  cluster states. In the 1D case, the vertices of the underlying hypergraph are arranged in a chain and labeled by natural numbers; all hyperedges have the form $\{j,j+1,j+2\}$ with $j\geq1$ and $j\leq n-2$, assuming $n\geq 3$. If we use $0,1,2$ to denote three colors, then the hypergraph can be colored by assigning vertex $j$ with the color $(j\mod 3)$. Similar analysis applies to 2D and higher-dimensional lattices. For simplicity, here we focus on the 1D case, which yields 
\begin{equation}
r_j=\begin{cases}
1 & n=3\mbox{ or } j=1 \mbox{ or } j=n,\\
2 &n\geq 4, j=2 \mbox{ or } j=n-1,\\
3 &  j\neq 1,2,n-1,n.
\end{cases}
\end{equation}
Therefore,
\begin{equation}
\sum_j 2^{r_j}=\begin{cases}
6 & n=3,\\
8n-20 &n\geq 4,
\end{cases}
\end{equation}
which implies that
\begin{equation}
N_{\mth}\geq \begin{cases}
\frac{18\ln\delta^{-1}}{\epsilon^2} & n=3,\\
\frac{8(2n-5)^2\ln\delta^{-1}}{\epsilon^2} &n\geq 4.
\end{cases}
\end{equation}

Next, consider the Union Jack state on the Union Jack chain; cf.~\fref{fig:hypergraph} in the main text. In this case, we have $r_j=2$ when $j$ corresponds to one of the four corners and $r_j=4$ otherwise.  Therefore,
\begin{equation}\label{eq:MTHboundUJchain}
\sum_j 2^{r_j}=16n-48,\quad
N_{\mth}\geq
\frac{128(n-3)^2\ln\delta^{-1}}{\epsilon^2}.
\end{equation}
Finally, consider the Union Jack state on the Union Jack lattice with $\tilde{n}\times \tilde{n}$ cells and   $n=\tilde{n}^2+(\tilde{n}+1)^2$ qubits. Calculation shows that
\begin{align}
\sum_j 2^{r_j}&=2^8(\tilde{n}-1)^2+2^4[\tilde{n}^2+4(\tilde{n}-1)]+2^2\times 4\nonumber\\
&=16(17\tilde{n}^2-28\tilde{n}+13 ),
\end{align}
which implies  that
\begin{align}
N_{\mth}&\geq
\frac{128(17\tilde{n}^2-28\tilde{n}+13 )   ^2\ln\delta^{-1}}{\epsilon^2}.\label{eq:MTHboundUJlattice}
\end{align}

\subsection{\label{sec:HHprotocol}Comparison with \rcite{HayaH18}}

Here, in the adversarial setting,
we compare our method with the method proposed by Hayashi and Hajdu\v{s}ek (HH)  \cite{HayaH18}, who considered the verification of graph states, but not hypergraph states. In addition, HH mainly focused on the case in which the graph is 3-colorable. They mentioned the general case briefly, but did not analyze the performance of their protocol in detail.
Since the main focus of \rcite{HayaH18} is self-testing,
HH do not trust their measurement devices.
However, after the verification of their measurement devices,
they verify their graph state under the assumption that
their measurement devices are trusty.

Suppose $|G\>$ is a
graph state associated with the graph $G$. When
$G$ is $m$-colorable, HH (Appendix~F of \rcite{HayaH18}) proposed the following  verification protocol, which consists of $m$
stabilizer tests.
Given a coloring  $\scrA=\{A_1, A_2, \ldots, A_m \}$  of $G$ with  $m$ colors, the verifier asks the adversary to prepare  $N+1$ systems with $N=m N'$.
After a random permutation of the $N+1$ systems,   $N$ systems are chosen and divided into $m$ groups  each with $N'$ systems. Then all systems in the $l$th group for $l=1,2,\ldots, m$ are subjected to the stabilizer test with $P_l$ [cf.~\eref{eq:PassProj} in the main text] as the projector onto the pass eigenspace. Let $\sigma$ be the reduced state of the remaining system after all these tests are passed.
If the $l$th test $P_l$ is passed  with significance level $\delta' $, then
one can guarantee that
$\tr[ \sigma (\id-P_l)]\le \frac{1}{\delta'(N'+1)}$.
If all the tests $P_1, \ldots, P_m$ are passed, with significance level $\delta:=m\delta' $,
then one can guarantee that
\begin{align}
\epsilon&=\tr[ \sigma (\id-|G\>\< G|)]
\le \sum_{l=1}^m \tr[ \sigma (\id-P_l)]
\nonumber\\
&\le \sum_{l=1}^m \frac{1}{\delta'(N'+1)}
=\frac{m^2}{\delta(N/m+1)}
\cong \frac{m^3}{\delta N }.
\end{align}
To verify $|G\>$ within infidelity $\epsilon$ and significance level $\delta$ in the adversarial scenario, the HH protocol requires about
$\lceil m^3/(\delta\epsilon)\rceil$ tests.

Now, we explain how our protocols outperform the HH protocol.
If we employ the coloring protocol and randomly choose the $l$th measurement setting with probability $1/m$,
then the verification operator has spectral gap $\nu(\Omega)=1/m$ according to
\thref{thm:CoverEfficiency}. If $N$ tests are passed with significance level $\delta$, then  we can guarantee that
\begin{align}
\epsilon=\tr [\sigma (\id-|G\>\< G|)]
\le \frac{m(1-\delta)}{N \delta }
\end{align}
according to \rscite{ZhuH19AdS,ZhuH19AdL}. 
So  the coloring protocol requires only
$\lceil m(1-\delta)/(\delta\epsilon)\rceil $ tests to verify $|G\>$ within infidelity $\epsilon$ and significance level $\delta$ in the adversarial scenario; cf. \eref{eq:NumTestAdvHG} in the main text. This is much more efficient than the HH protocol. Thanks to \eref{eq:NumTestAdvHGhcolor}, the  hedged  coloring protocol can further reduce the number of tests and achieve the optimal scaling behavior with $\delta$.  When $m=3$ and $\epsilon=\delta=0.01$  for example, the HH protocol in \rcite{HayaH18} requires 270000 tests, while the hedged coloring protocol requires only 1870 tests, which is smaller by 144 times.

\subsection{\label{sec:TakeM18}Comparison with \rcite{TakeM18}}
Recently,  Takeuchi and  Morimae (TM) \cite{TakeM18} introduced a protocol for verifying general hypergraph states whose orders are upper bounded by a constant. Recall that the order of a hypergraph $G=(V,E)$ is the maximum cardinality of hyperedges in the hyperedge set $E$.

Let $G=(V,E)$ be a hypergraph such that $2\leq |e|\leq c$ for all $e\in E$, where $c$ is a positive constant. Let $k\geq (4n)^7$ and $m\geq (2\ln2)n^3k^{18/7}$ be positive integers.
According to Theorem~5 in \rcite{TakeM18},
to verify the  hypergraph state $|G\>$ within infidelity $\epsilon=k^{-1/7}$ and significance level $\delta=k^{-1/7}$, the number of tests required by the TM protocol is given by
\begin{align}
N_{\mathrm{TM}}&=m+nk\geq (2\ln2) n^3k^{18/7}+nk> (2 \ln2) n^3k^{18/7} \nonumber \\
&=(2\ln 2) n^3\epsilon^{-18}. \label{eq:NTM1}
\end{align}
Note that the conditions $k\geq (4n)^7$ and $\epsilon=k^{-1/7}$ imply the inequality $4n\epsilon\leq1 $. For example, when $k= (4n)^7$ and $\epsilon=\delta=k^{-1/7}=1/(4n)$, the number of required tests satisfies
\begin{align}
N_{\mathrm{TM}}&\geq  (2\ln2)n^3(4n)^{18}
+n(4n)^7=(2^{37}\ln2)n^{21}+2^{14}n^8\nonumber\\
&>(2^{37}\ln2)n^{21}>9.5\times 10^{10}n^{21}.
\end{align}
Although
this number is still polynomial in $n$, it is already astronomical in the simplest nontrivial scenario with $n=3$.  So it is too prohibitive to apply the TM protocol in any scenario of practical interest.
By contrast,
the number of tests required by our coloring protocol satisfies
\begin{align}
N&\leq (16n^2-4n)\chi(G)< 16n^2\chi(G)\leq 16n^3
\end{align}
according to \eref{eq:NumTestAdvHG2} in the main text,
which is dramatically smaller than $N_{\mathrm{TM}}$. The hedged coloring protocol can further reduce the number of tests according to  \esref{eq:NumTestAdvHGhcolor} and \eqref{eq:NumTestAdvHGhcolorOpt}.

Our protocols are not only much more efficient than the TM protocol, but also much simpler to apply. In particular, the TM protocol relies on adaptive stabilizer tests, while
our protocols do not rely on any adaption. In addition, the data processing in the TM protocol is a bit involved, while it is very simple in our protocols. Furthermore, TM did not derive the explicit number of required tests except  for restricted choices of the infidelity $\epsilon$ and significance level $\delta$, which makes it difficult to apply their result in many scenarios of practical interest. By contrast, we  derive the explicit number of required tests  for all valid choices of $\epsilon$ and  $\delta$.

\bigskip

\subsection{Comparison with \rcite{PallLM18}}

Recall that graph states are special hypergraph states in which all edges have order 2. It is known that all stabilizer  states are equivalent  to graph states under local Clifford transformations \cite{Schl02,GrasKR02}.
Due to their simple structures, stabilizer states can be verified efficiently using Pauli measurements \cite{FlamL11,PallLM18}. The protocol proposed by
Pallister, Linden, and  Montanaro (PLM) \cite{PallLM18} is particularly efficient with respect to  the total number of tests.
To be specific, the PLM protocol measures all $2^n-1$ nontrivial stabilizer operators of $|G\>$ in the Pauli group with equal probability.
The resulting verification operator  has the form
\begin{align}\label{eq:OmegaPLM}
\Omega_{\rm PLM} = |G\>\<G|+ \frac{2^{n-1}-1}{2^n-1}(\id-|G\>\<G|),
\end{align}
with
\begin{equation}\label{eq:vPLM}
\beta(\Omega_{\rm PLM})=  \frac{2^{n-1}-1}{2^n-1},\quad   \nu(\Omega_{\rm PLM})= \frac{2^{n-1}}{2^n-1}.
\end{equation}
To verify $|G\>$ within infidelity $\epsilon$ and significance level $\delta$, this protocol requires about
\begin{equation}
\lceil2^{1-n}(2^n-1)\epsilon^{-1}\ln\delta^{-1}\rceil\leq\lceil 2\epsilon^{-1}\ln\delta^{-1}\rceil
\end{equation}
tests, which is  smaller than the number $\lceil\chi(G)\epsilon^{-1}\ln\delta^{-1}\rceil$ required by our coloring protocol [cf.~\eref{eq:NumTestCoverUB}]. However,   the number of potential measurement settings required by the PLM protocol
increases exponentially with the number $n$ of qubits. When $n$ is large, this protocol will be impractical  if it is costly or time consuming to switch measurement settings. By contrast, our coloring protocol requires at most $n$ potential measurement settings. In addition,
when the chromatic number $\chi(G)$ of $G$ is small (in particular when $G$ is 2-colorable), the total number of tests required   is comparable to the PLM protocol.
Furthermore, the PLM protocol requires  $Y=\rmi XZ$ measurement because
it is necessary to measure all nontrivial stabilizer operators of $|G\>$, while
our protocol requires only  $X$ and $Z$ measurements.

Incidentally, \rcite{PallLM18} introduced another protocol for verifying the graph state $|G\>$ by measuring  $n$ stabilizer generators of $|G\>$ with equal probability.  The resulting verification operator $\Omega$ has spectral gap $\nu(\Omega)=1/n$. This protocol requires $\lceil n\epsilon^{-1}\ln\delta^{-1}\rceil $ tests in total, which  corresponds to the performance of our coloring protocol in the  worst case in which the graph is complete (contains all possible edges).
In general,  the coloring protocol requires much fewer measurement settings and tests in total.

\subsection{Comparison with \rcite{TakeMMM19}}
Very recently, Takeuchi, Mantri, Morimae, Mizutani, and Fitzsimons (TMMMF)  introduced a protocol for  verifying graph states with a very small  significance level \cite{TakeMMM19} (posted on arXiv after our paper). To be specific, given a graph state $|G\>$ of $n$ qubits, by  performing
$N_{\rm TMMMF}=2n \lceil (5 n^4 \ln n)/32\rceil$ tests,  the protocol proposed in \rcite{TakeMMM19} guarantees that
the resultant state $\sigma$ satisfies the following condition
\begin{align}
\langle G|\sigma |G\rangle
\ge 1- \frac{2 \sqrt{c}+1}{n}
\end{align}
if these tests are passed
with significance level $n^{1-5c/64}$.
Here, $c$ is a constant that satisfies $\frac{64}{5}< c< \frac{(n-1)^2}{4}$.

Next, we consider the performance of the hedged coloring  protocol proposed in the main text.  Suppose the graph $G$ is $m$ colorable and we apply a hedged coloring protocol with $m$ colors. By  \eref{eq:NumTestAdvHGhcolor}, to verify $|G\>$ within infidelity $\epsilon=\frac{2 \sqrt{c}+1}{n}$ and significance level $\delta=n^{1-5c/64}$,   the number of  required tests reads
\begin{align}
N&=\biggl\lfloor \frac{h_*(1/m) \ln(F\delta)^{-1}}{\epsilon}\biggr\rfloor\leq \frac{(m+\rme-1)\ln(F\delta)^{-1}}{\epsilon},\nonumber\\
&\approx \frac{(m+\rme-1)(\frac{5c}{64}-1)n \ln n}{2 \sqrt{c}+1}\leq O(n^2 \ln n),
\end{align}
where $F=1-\epsilon$ and the approximation holds as long as $\epsilon,\delta\ll1$. For most graph states of practical interest, the chromatic numbers are upper bounded by a small constant, so
$N= O(n\ln n)$ tests are sufficient.  This number is much smaller than
$N_{\rm TMMMF}$. Therefore, our hedged  coloring protocol is much  more efficient that the protocol presented in  \rcite{TakeMMM19}.

%%%%%%%%%%%%%%%%%%%%%%%%%%%%%%%%%%%%%%%%%%%%%%%%%%%%%%%%%%%%%%%%%%%%%%%%%%%%%%%%%%%%%%%%

%@CONTROL{REVTEX41Control}
%@CONTROL{apsrev41Control,author="48",editor="1",pages="1",title="0",year="0"}

\nocite{apsrev41Control}
\bibliographystyle{apsrev4-1}
\bibliography{all_references}

\end{document}